\documentclass[11pt]{article}
\usepackage[T1]{fontenc}
\usepackage{graphicx}
\usepackage{fullpage}
\usepackage{xcolor}
\usepackage{amsmath,amssymb,amsfonts,amsthm,enumerate}
\usepackage{nicefrac}
\usepackage{hyperref,cleveref}
\usepackage{authblk}
\usepackage{times}
\usepackage{xspace}
\usepackage{bbm}
\usepackage[margin=1in]{geometry}
\usepackage[procnumbered, linesnumbered,algoruled,noend]{algorithm2e}
\usepackage{multirow}
\usepackage{amsmath, amsthm, amsfonts, mathdots, thm-restate, mathrsfs, mathtools}
\usepackage{enumitem}
\usepackage{amsmath,amssymb,amsthm,mathtools}
\usepackage[margin=1in]{geometry}
\usepackage{enumitem}
\usepackage[normalem]{ulem}

\usepackage[
    backend=biber,
    style=numeric,
    sortcites=true,
    maxcitenames=99,
    maxbibnames=99
]{biblatex}
\usepackage[margin=1in]{geometry}
\usepackage{hyperref}
\usepackage{cleveref}

\crefname{algocf}{algorithm}{algorithms}
\Crefname{algocf}{Algorithm}{Algorithms}

\newtheorem{theorem}{Theorem}[section]

\newtheorem{claim}[theorem]{Claim}

\newtheorem{lemma}[theorem]{Lemma}

\newtheorem{definition}[theorem]{Definition}

\newcommand{\E}{\mathbb{E}}
\newcommand{\RR}{\mathbb{R}}
\newcommand{\OPT}{O}
\newcommand{\1}{\mathbf{1}}
\newcommand{\Ical}{\mathcal{I}}

\newcommand{\Mcal}{\mathcal{M}}

\newcommand{\Bcal}{\mathcal{B}}
\newcommand{\indicator}{\mathbf{1}}
\newcommand{\vc}[1]{\mathbf{#1}}
\newcommand{\comment}[1]{}
\usepackage{xcolor}
\newcommand{\swap}{\textnormal{\texttt{swap}}}
\DeclareMathOperator*{\argmax}{arg\,max}

\newcommand{\ignore}[1]{}
\def \M {\mathcal M}
\def \I {\mathcal I}
\def \F {\mathcal F}
\def \cons {4}
\def \constwo {20}
\def \Ev {\mathcal{E}}
\newcommand{\mar}{\textnormal{Mar}}

\newcommand*{\email}[1]{\href{mailto:#1}{\nolinkurl{#1}} }

\newenvironment{claimproof}
  {%
    \begin{proof}%
  }
  {%
    \end{proof}%
  }

\definecolor{thiagocolor}{RGB}{128, 0, 128}
\newcommand{\thiago}[1]{{\color{thiagocolor} \sf $\infty$ Thiago: [#1]}}

\definecolor{mohitcolor}{RGB}{0, 128, 128}
\newcommand{\ms}[1]{{\color{mohitcolor} \sf $\infty$ Mohit: [#1]}}

\newcommand{\roy}[1]{{\color{magenta} \sf $\infty$ Roy: [#1]}}

\newcommand{\sgsplong}{{Spiteful Greedy Swap Poisson Process}\xspace}
\newcommand{\sgsp}{{SGS-Poisson}\xspace}

\newcommand{\sm}{{Sub-Mat}\xspace}
\newcommand{\sw}{{Submodular Welfare}\xspace}

\title{If it is Good Then Drop it -- a Spiteful Poisson Process for Submodular Maximization}

\author{Ariel Kulik}
\affil{
  {Ben-Gurion University of the Negev},
  {Beer-Sheva},
  {Israel}.
\email{kulik@bgu.ac.il}
}
\author{Thiago Oliveira
\thanks{Supported in part by NSF awards CCF-2440113, 2504994, 2106444 and the 2026 Spring ARC-ACO Fellowship.}}
\affil{%
  { H. Milton Stewart School of Industrial and Systems Engineering, Georgia Tech},
  {Atlanta},
  {USA}.
\email{toliveira9@gatech.edu}
}
\author{Roy Schwartz}
\affil{%
  {Technion},
  {Haifa},
{Israel}.
\email{schwartz@cs.technion.ac.il}
}
\author{Mohit Singh\thanks{Supported in part by NSF awards CCF-2504994, 2106444.}}
\affil{%
  {H. Milton Stewart School of Industrial and Systems Engineering, Georgia Tech},
  {Atlanta},
  {USA}.
\email{mohit.singh@isye.gatech.edu}
}\date{}

\begin{document}
\maketitle


\abstract{We study the problem of maximizing a general and not necessarily monotone submodular function subject to a matroid independence constraint.
This problem has a rich history, with multiple algorithms using both discrete and continuous methods.
Recently, [Ganz-Rozenman, Kulik, Schwartz and Singh STOC `26] presented a novel hybrid approach based on a Poisson process that aims to combine the strengths of both discrete and continuous methods for the special case of the problem where the submodular function is monotone.

Our main result is a new Poisson process based hybrid algorithm that works for both non-monotone and monotone submodular functions, achieving an approximation of $ \nicefrac{1}{e}$ for the former and $1-\nicefrac{1}{e}$ for the latter.
The algorithm always maintains a feasible set and at random times governed by the Poisson process it performs a single element swap based on a best response set.
The new idea is that our algorithm is spiteful as it can purposefully discard an element that is in both the current set and the best response set.
Surprisingly, this spiteful step does not harm the approximation our algorithm achieves for monotone submodular functions but is necessary for the non-monotone case.
As applications, we obtain fast approximation algorithms for maximizing non-monotone submodular function subject to a general matroid independence constraint as well as faster algorithms for a partition matroid.}

\section{Introduction}\label{sec:introductions}
Submodular maximization is a classic topic in the theory of algorithms and combinatorial optimization that has been studied for several decades starting from the late $70$'s~\cite{NW78,nemhauser1978analysis,fisher1978analysis}.
A function $ f\colon 2^U\rightarrow \RR_+$ is submodular if it has the diminishing returns property:
$ f(A+i)-f(A)\geq f(B+i) - f(B)$ for every $A\subseteq B\subseteq U$ and $ i\in U\setminus B$.\footnote{For simplicity of presentation $A+i$ and $ A-i$ denote $A\cup \{ i\}$ and $ A\setminus \{ i\}$, respectively, for every $A\subseteq U$ and $i\in U$.}
Submodular functions are a rich class of functions that naturally arise in various settings, e.g., graph and hypergraph cut functions in combinatorics, rank functions in linear algebra, and entropy in probability, to name a few.
Thus, it is no surprise that submodular maximization has found many real-world applications, including but not limited to: influence maximization in social networks~\cite{kempe2003maximizing,HMS08,MR10}, sensor placement and data summarization~\cite{KLGVF08,KSG08,LB10,LB11,mirza16}, and feature selection in machine learning~\cite{LWKSB13,KEDNG17,BHZ22} (see also, e.g., a more comprehensive survey~\cite{B13}).
Moreover, submodular maximization has been studied in various computational settings, including parallel and distributed computation~\cite{BS18,BRS19,CQ19,CK19b,DENW16}, online computation~\cite{GNS23,FZ18,BFS19online,BFFG20,BHZ13,KPV13,KMZ18}, and streaming~\cite{AEFNS22,CGQ15,FLNSZ22,FNSZ23,KMZLK19}.

In this work, the problem of maximizing a submodular function $f$ subject to a matroid $\Mcal=(U,\Ical) $ independence constraint is considered and is denoted by \sm: $ \max \{ f(S)\colon S\in \Ical\}$. 
\sm has been the focus of intense research in the last several decades, as it captures classic algorithmic problems, e.g., maximum cut in undirected and directed graphs~\cite{GW95,Karp72,KKMO07,HZ01,LLZ02,,BHPZ26}, facility location~\cite{AS99,CFN77,CFN77b}, as well as generalized assignment~\cite{CCPV11,fleischer2006tight,feige2006approximation} and maximum $k$-coverage~\cite{F98,KMN99}.

A long sequence of works~\cite{LMNS10,V13,OV11,CVZ14,FNS11,EneN16a,BuchcinderF16a,BuchbinderF23a,BFNS14,GNS23} has focused on \sm with a general non-negative submodular objective $f$ that is not necessarily monotone.\footnote{$f$ is monotone if $ f(A)\leq f(B)$ for every $ A\subseteq B\subseteq U$.}
These works are based on a combination of discrete techniques, e.g., local search and greedy, as well as continuous techniques. 
A noteworthy algorithm is the measured continuous greedy algorithm of Feldman, Naor and Schwartz~\cite{FNS11} that is based on the seminal continuous greedy algorithm of Calinescu, Chekuri, Pal and Vondr\'{a}k~\cite{CCPV11} for \sm with a monotone objective $f$.
The measured continuous greedy algorithm achieves an approximation of $ \nicefrac{1}{e}$ for \sm. 
Moreover, all subsequent improvements (including the current best known algorithms) use the measured continuous greedy algorithm as a central building block: Ene and Nguyen~\cite{EneN16a} present an improved guarantee of $ \approx \nicefrac{1}{e}+0.004$, Buchbinder and Feldman~\cite{BuchcinderF16a} further push the approximation to $ \approx \nicefrac{1}{e}+0.017$, while Buchbinder and Feldman~\cite{BuchbinderF23a} provide the current best known approximation of $ \approx \nicefrac{1}{e}+0.033$.  

In the continuous approach, one approximately solves a non-convex relaxation to \sm defined via the multilinear extension $F\colon [0,1]^U\rightarrow \RR_+$ that is defined as:
$F(\mathbf{x})\triangleq \sum _{S\subseteq U}f(S)\cdot \Pi _{i\in S}x_i \cdot \Pi _{i\notin S}(1-x_i) $.
The multilinear extension carries an intuitive probabilistic interpretation of $ F(\mathbf{x})=\mathbb{E}_{R\sim \mathbf{x}}[f(R)]$, where $ R\sim \mathbf{x}$ denotes a random subset in which every element $i\in U$ is independently chosen to $R$ with a probability of $x_i$.

While powerful, it is known that the above continuous approach has several drawbacks (see, e.g.,~\cite{GanzKSS26a}).
First, the implementation of a continuous time process as an algorithm requires discretization.
Therefore, to bound the error this discretization incurs, sufficiently small steps are needed, typically resulting in a large number of iterations.
Second, the continuous approach usually requires numerous evaluations of the gradient of $F$, where each such evaluation is performed by sampling values of $f$, which can be computationally expensive.
Third and last, rounding of the fractional solution found by the continuous algorithm is required.

In addition to the continuous approach, discrete techniques, e.g., local search~\cite{LMNS10} and randomized greedy~\cite{BFNS14,GNS23}, were applied to \sm with a non-monotone objective $f$.
While these approaches fail to obtain the $\nicefrac{1}{e}$ guarantee of the measured continuous greedy algorithm (\cite{LMNS10} and~\cite{GNS23} achieve approximations of $\nicefrac{1}{4}$ and $\approx 0.274$ respectively), they benefit from two facts.
First, no discretization of a continuous time process is required.
Second, these approaches work directly with $f$ and therefore do not require gradient computation of the multilinear extension $F$ or rounding, leading to simpler and faster algorithms.

Recently, Ganz-Rozenman, Kulik, Schwartz and Singh~\cite{GanzKSS26a}, introduced a hybrid approach that aims to combine the strengths of both continuous and discrete approaches, resulting in a Poisson process based algorithm for \sm with a monotone objective $f$.
The main contribution of~\cite{GanzKSS26a} is conceptual: providing a new method for designing good algorithms for \sm with a monotone objective $f$.
Moreover, byproducts of this hybrid approach are fast approximation algorithms for several special cases of \sm with a monotone objective $f$, perhaps most notable is the special case where $\Mcal$ is a partition matroid.

Unfortunately, the Poisson process algorithm of~\cite{GanzKSS26a} heavily relies on the fact that the submodular objective $f$ is monotone.
Thus, extending this hybrid approach to a general submodular objective $f$ that is not necessarily monotone is an interesting challenge and the main focus of this work.

\subsection{Our Results}\label{sec:results}

Our main contribution is conceptual, as we prove that the Poisson process based hybrid approach can also be applied to non-monotone submodular objectives $f$.

Inspired by \cite{GanzKSS26a}, we present a new stochastic process which we denote as \sgsplong (\sgsp).
The internal state of \sgsp consists of a single independent set $ A\in \Ical$ that is initialized to an empty set.
Starting from time $ t=\varepsilon$, a non-homogeneous Poisson process with rate $\lambda(t)=\nicefrac{k}{t}$, where $k$ is the rank of the matroid $\Mcal$, is executed.
If the next event of the Poisson process occurs at time $t$: 
\begin{enumerate}
    \item A base $ Z(t)\in \Bcal$ that maximizes $ \sum _{j\in Z(t)}\left( F(t\indicator_{A}\vee\indicator_{\{ j\}})-F(t\indicator_{A})\right)$ is computed.\footnote{Standard arguments imply that it can be assumed without loss of generality that $Z(t)$ is a base by adding $k$ dummy elements to $U$ that contribute zero to $f$.}
    \item A map $h\colon Z(t)\rightarrow A\cup \{\perp \}$ is constructed such that:
    (a) for every $ j\in Z(t)\cap A$: $h(j)=j$; (b) for every $ j\in Z(t)$: $ A-h(j)+j\in \Ical$; and (c) for every $i\in A$ there is exactly a single $ j\in Z(t)$ such that $ h(j)=i$.
    ($\perp \notin U$ denotes a \emph{non-element} and thus $ A-\perp = A$).
    \item A single uniform random element $ j\in Z(t)$ is selected and:
    \begin{enumerate}
        \item $h(j)$ is swapped with $j$: $ A\leftarrow A-h(j)+j$.
        \item if $ h(j)=j$ then with probability $t$,  $h(j)$ is dropped: $ A\leftarrow A-h(j)$.\label{sgsp-drop}
    \end{enumerate}
\end{enumerate}
The above steps are repeated every time an event occurs, stopping at time $t=1$.

The novelty of \sgsp lies in step \eqref{sgsp-drop}, where an element that is both in the current solution $A$ and the best base $Z(t)$, is dropped with a probability of $\nicefrac{t}{k}$ if an event of the Poisson process occurs at time $t$: the element is selected with a probability of $\nicefrac{1}{k}$ and then dropped with a probability of $t$.
Hence, the random process is \emph{spiteful} since it can purposefully drop good elements, i.e., elements in $Z(t)$, that are already in the solution $A$ at time $t$.
Surprisingly, this spiteful step not only plays a major role in coping with a non-monotone objective $f$, but it also does not harm the approximation one can obtain when $f$ is monotone (see Theorem~\ref{thrm:Poisson} below).

One can note that \sgsp retains the same key features of the hybrid approach of~\cite{GanzKSS26a}.
First, there is no discretization since the Poisson process can be simulated by directly sampling the time of the next event.
Second, no rounding is required since the algorithm operates directly on independent sets.
Third and last, the expected number of single element operations (swapping two single elements, adding a single element, and removing a single element) is very low and at most only $ k\ln{(\nicefrac{1}{\varepsilon})}$.

The following theorem states the approximation \sgsp achieves for \sm with a non-monotone and monotone objective $f$.
\begin{theorem}\label{thrm:Poisson}
Given a submodular function $ f\colon 2^U\rightarrow \RR_+$, a matroid $ \Mcal=(U,\Ical)$,  and a starting time $ 0<\varepsilon \leq 1$, \sgsp finds $ A\in \Ical$ satisfying $ \mathbb{E}[f(A)]\geq (1-\varepsilon)(\nicefrac{1}{e}) f(O)$ in case $f$ is non-monotone and $ \mathbb{E}[f(A)]\geq (1-\varepsilon)(1-\nicefrac{1}{e})f(O)$ in case $f$ is monotone.
Here $O=\argmax \{ f(S):S\in \Ical\}$ is an optimal independent set.
\end{theorem}

There are two things to note regarding Theorem~\ref{thrm:Poisson}.
First, \sgsp works simultaneously for both non-monotone and monotone objectives $f$, achieving for the latter a tight approximation of $(1-\nicefrac{1}{e})$ (thus matching the tight known algorithms of~\cite{CCPV11,filmus2014monotone,BF24deter} as well as the recent Poisson process based algorithm of~\cite{GanzKSS26a}).
Second, \sgsp does not improve upon the best known approximation of $\approx \nicefrac{1}{e}+0.033$ for \sm with a non-monotone objective $f$ given by Buchbinder and Feldman~\cite{BuchbinderF23a} that is based on a continuous approach in which the measured continuous greedy algorithm is a key ingredient.
However, \sgsp improves upon all known discrete algorithms for \sm, i.e., local search~\cite{LMNS10} which achieves an approximation of $ \nicefrac{1}{4}$ and random greedy~\cite{BFNS14,GNS23} which achieves an approximation of $ \approx 0.274$, while matching the $\nicefrac{1}{e}$ guarantee of the measured continuous greedy itself~\cite{FNS11}.

\paragraph{Applications.}
Our applications include fast approximation algorithms for \sm for general and partition matroids.
Recall that a partition matroid $ \Mcal = (U,\Ical)$ is associated with a partition $(U_1,\ldots,U_k)$ of $U$ and a set $ S\in \Ical$ if and only if $ |S\cap U_j|\leq 1$ for every $1\leq j\leq k $.
The special case of \sm with a partition matroid is of special interest since it captures, e.g., the \sw problem with non-monotone utilities (see, e.g.,~\cite{FNS11,GNS23} and the references therein).

In what follows, the running time of an algorithm is measured by the number of function evaluations, i.e., value oracle queries, it performs (as this is the standard measure for these types of problems).
We assume value oracle access to either $f$ or its multilinear extension $F$ and present results for both.
Moreover, to abbreviate presentation the size of the universe $U$ is denoted by $n$ and $O=\argmax \{ f(S):S\in \Ical\}$ is an optimal independent set.

The following two theorems summarize the fast approximation algorithms.
\begin{theorem}\label{thrm:fastF}
Given a submodular function $ f\colon 2^U\rightarrow \RR_+$, a matroid $ \Mcal=(U,\Ical)$, and $ 0<\varepsilon \leq 1$, there exists an algorithm that performs in expectation $ O(nk\ln(\nicefrac{1}{\varepsilon}))$ evaluations of $F$ and finds $ A\in \Ical$ satisfying: $ \mathbb{E}[f(A)]\geq (1-\varepsilon)(\nicefrac{1}{e})f(O)$.
Moreover, if $\Mcal$ is a partition matroid, then there exists an algorithm achieving the same approximation that performs in expectation $ O(n\ln{(\nicefrac{1}{\varepsilon})})$ evaluations of $F$.
\end{theorem}

\begin{theorem}\label{thrm:fastf}
Given a submodular function $ f\colon 2^U\rightarrow \RR_+$, a matroid $ \Mcal=(U,\Ical)$, and $ 0<\varepsilon \leq 1$, there exists an algorithm that performs in expectation $O\left(\frac{nk^2}{\varepsilon^2}\ln^2{(\nicefrac{1}{\varepsilon})}\ln n\right)$
evaluations of $f$ and finds $ A\in \Ical$ satisfying: $ \mathbb{E}[f(A)]\geq (1-\varepsilon)(\nicefrac{1}{e})f(O)$.
Moreover, if $\Mcal$ is a partition matroid, then there exists an algorithm achieving the same approximation that performs in expectation $O\left(\frac{nk}{\varepsilon^2}\ln^2{(\nicefrac{n}{\varepsilon})}\right)$ evaluations of $f$.  
\end{theorem}

Let us compare our fast approximation algorithm for \sm when $\Mcal$ is a general matroid and the algorithm has value oracle access to $f$ with known results.
As previously noted, algorithms achieving an approximation strictly better than $\nicefrac{1}{e}$ are known~\cite{EneN16a, BuchcinderF16a,CGLXZ26, BuchbinderF23a}.
However, all these algorithms perform a large polynomial in $n$ number of value queries to $f$, estimated in the range from $O(n^5)$ to $O(n^{11})$.
Fast approximation algorithms, achieving the same approximation of $ (1-\varepsilon)(\nicefrac{1}{e})$ as in Theorem~\ref{thrm:fastf} are known, where the algorithm of Buchbinder and Feldman~\cite{BuchbinderF24a} performs $O(2^{O(1/\varepsilon^4)}n\ln{k})$ evaluations and the algorithm of Segui-Gasco and Shin~\cite{Segui-GascoS17a} performs $O(n^3k^2\ln^2{(\nicefrac{n}{\varepsilon})}/\varepsilon^4)$ evaluations in the worst case.
It is important to note, when considering the number of value oracle queries to $f$, that the former has a large exponential dependence on $ \nicefrac{1}{\varepsilon}$ and the latter is strictly slower than the guarantee of Theorem~\ref{thrm:fastf}.

Faster algorithms are known, however those return weaker approximation guarantees, e.g., Ganz, Nuti and Schwartz~\cite{GNS23} provide a weaker approximation of $ \approx 0.274$ using $ O(nk)$ evaluations and Han, Cao, Cui and Wu~\cite{HanCC20a} provide a weaker approximation of $ (\nicefrac{1}{4}-\varepsilon)$ using $ O(n\ln{(\nicefrac{k}{\varepsilon})}/\varepsilon)$ evaluations.

\subsection{Our Techniques}\label{sec:techniques}
\paragraph{\sgsplong.}
To present the intuition behind \sgsp, and specifically its spiteful step~\eqref{sgsp-drop}, we first briefly sketch the measured continuous greedy algorithm~\cite{FNS11}.
The reason is that \sgsp can be seen as the limit of rounding on the fly of the measured continuous greedy algorithm scaled by time when its step size $\delta$ approaches zero.

The inner state of the measured continuous greedy algorithm at time $t$ is $\mathbf{x}(t)$, where:
$(1)$ $\mathbf{x}(0)\leftarrow \mathbf{0}$; and
$(2)$ for every $ 0< t\leq 1$: let $ Z(t)\triangleq \argmax \{ \sum _{i\in Z}(F(\mathbf{x}(t)\vee \indicator _{\{ i\}})-F(\mathbf{x}(t)))\colon Z\in \Ical\}$ and set $ x_i(t+\delta)\leftarrow x_i(t)+\delta (1-x_i(t))$ for every $ i\in Z(t)$ and $ x_i(t+\delta)\leftarrow x_i(t)$ for every $ i\in U\setminus Z(t)$
(recall that it can be assumed without loss of generality that $Z(t)$ is a base).

Assume we want to round on the fly when moving from time $t$ to time $t+\delta$.
Specifically, $\mathbf{x}(t)$ is already rounded, i.e., $ \mathbf{x}(t)=t\indicator _{A(t)}$ for some random $ A(t)\in \Ical$ that denotes the inner state of \sgsp at time $t$.
The goal is to round $\mathbf{x}(t+\delta)$ into $(t+\delta)\indicator _{A(t+\delta)} $ for some $ A(t+\delta)\in \Ical$.
Note that since $ \mathbf{x}(t)=t\indicator _{A(t)}$ is already rounded, the definition of the measured continuous greedy algorithm gives:
\begin{align}
    \mathbf{x}(t+\delta) = t\indicator _{A(t)} + \delta \indicator _{Z(t)\setminus A(t)} + \delta (1-t)\indicator _{Z(t)\cap A(t)}. \nonumber 
\end{align}
Applying the randomized variant of pipage rounding~\cite{AS04,CCPV11} to $ \mathbf{x}(t+\delta)/(t+\delta)$, i.e., $\mathbf{x}(t+\delta)$ scaled by the time $t+\delta$, boils down to the following basic three operations:
\begin{enumerate}
\item w.p. $\approx 1-\delta |Z(t)\setminus A(t)|/t-\delta|Z(t)\cap A(t)|$ set $ A(t+\delta)$ to be $A(t)$.\label{rounding1}
\item w.p. $\approx \delta |Z(t)\setminus A(t)|/t$ swap a uniform random element $ j\in Z(t)\setminus A(t)$ with $h(j)$.\label{rounding2} 
\item w.p. $\approx \delta |Z(t)\cap A(t)|$ drop a uniform random element $ j\in Z(t)\cap A(t)$.\label{rounding3} 
\end{enumerate}
There are two important things to note.
First, operation~\eqref{rounding3} above directly follows from the $ (1-x_i(t))$ slowdown  in the definition of the measured continuous greedy algorithm.
Second, operation~\eqref{rounding3} above corresponds to the spiteful step~\eqref{sgsp-drop}, where each element of $ Z(t)\cap A(t)$ is dropped with a probability of $t$ conditional on it being chosen from $Z(t)$.

It is trivial to observe that if operation~\eqref{rounding1} is chosen, then from an algorithmic perspective nothing needs to be done (as $ A(t+\delta)$ is set to $A(t)$).
Since the random variant of pipage rounding is oblivious to the objective, all that is needed is to understand the time of the next swap, drop, or addition, in this continuous process as $ \delta$ approaches $0$.
Calculating these limits, one can show that this is exactly given by a Poisson process with rate $ \lambda(t)=\nicefrac{k}{t}$ (as in \sgsp) that we analyze directly.

\paragraph{Applications.}
Fast approximation algorithms with value oracle access to the multilinear extension $F$ follow rather directly from the simplicity of \sgsp: at each event of the Poisson process, the algorithm only needs to identify a suitable swap and perform one single-element update.

The main challenge is to obtain a fast implementation when the algorithm has access only to the value oracle of $f$.
To achieve this, two additional ingredients are required.

First, we introduce a general swap procedure that captures exactly the properties needed by the analysis of \sgsp.
This abstraction is stronger than merely choosing an approximately optimal best-response element (as was previously done by~\cite{GanzKSS26a} for the case the objective $f$ is monotone).
The reason for this is that the analysis for a non-monotone objective $f$ heavily relies on controlling the probability that each individual element appears in the solution.
Therefore, the swap procedure must also guarantee suitable upper bounds on the probability that any fixed element is selected.
This is the main distinction from the Poisson process based algorithm of~\cite{GanzKSS26a} for the case the objective $f$ is monotone, where weaker best-response guarantees are sufficient.

Second, we show how to implement such swap procedures using samples from $f$.
Since estimates of marginal values can have large variance on arbitrary instances, we first apply a preprocessing step, as is common in fast submodular maximization algorithms (see, e.g.,~\cite{BuchbinderFS14a,EN19}). 
We give a simplified analysis of this preprocessing and obtain tighter bounds via martingale arguments. After preprocessing, a small number of samples from $f$ suffices to mimic the required value oracle calls to the multilinear extension $F$, leading to a fast implementation of \sgsp with value oracle queries to $f$.

\subsection{Additional Related Work}\label{sec:relatedwork}
As previously mentioned, submodular maximization with a submodular objective $f$ that is not necessarily monotone, has been studied for more than two decades now via a variety of methods, including local search, randomized greedy and continuous approaches, and for various types of constraints.
Due to space constraints, only some of the most relevant works that were not previously discussed, are elaborated on (the reader is referred to a survey, e.g.,~\cite{Gonzalez2018HandbookV1} Vol. $1$ Chapter $42$, for a more comprehensive discussion).

The best known hardness result for \sm with an objective $f$ that is not necessarily monotone was given by Oveis Gharan and Vondr\'{a}k~\cite{OV11} for a partition matroid, who proved that any algorithm achieving an approximation better than $ \approx 0.478$ must query $f$ an exponential number of times.
This result was later extended by Qi~\cite{Qi23} to a uniform matroid, i.e., a cardinality constraint.

The above hardness stands in contrast to an approximation of $\nicefrac{1}{2}$ for the unconstrained problem that is obtained with only $ O(n)$ evaluations of $f$, that was given by Buchbinder, Feldman, Naor and Schwartz~\cite{BFNS15}.
The latter is known to be tight as a matching hardness of $ \nicefrac{1}{2}$ was given by Feige, Mirrokni and Vondr\'{a}k~\cite{FMV11}.
Another notable special case of \sm is when $\Mcal$ is a uniform matroid, i.e., a cardinality constraint.
In addition to all previously mentioned results for \sm, the following faster algorithms (that achieve worse approximations than the current best known) are given for a uniform matroid:
Tukan, Mualem and Feldman~\cite{TukanMF24a} achieve an approximation of $ \approx \nicefrac{1}{e}+0.017$ with $ O(n+k^2)$ evaluations of $f$;
Buchbinder, Feldman and Schwartz~\cite{BuchbinderFS14a} present an approximation of $ \nicefrac{1}{e}-\varepsilon$ with $ \min \{ n\ln{(\nicefrac{1}{\varepsilon})\varepsilon^2},n\ln{(\nicefrac{k}{\varepsilon})/\varepsilon+k\sqrt{n\ln{(\nicefrac{k}{\varepsilon})/\varepsilon}}}\}$ evaluations of $f$; and Sakaue~\cite{Sakaue20a} presents an approximation of $\nicefrac{1}{4}-\varepsilon$ with $ O(n+\nicefrac{k}{\varepsilon})$ evaluations of $f$.

It is worth noting that the recent Poisson process based algorithm for submodular maximization of Ganz-Rozenman, Kulik, Schwartz and Singh~\cite{GanzKSS26a} has already found applications, see, e.g.,~\cite{BFLS26} for its uses for streaming algorithms. 
Moreover, it is interesting to note that the use of a Poisson process in the design and analysis of algorithms was independently considered prior to~\cite{GanzKSS26a} in settings unrelated to submodular maximization, e.g., Makarychev, Schudy and Sviridenko~\cite{MSS12} used such an approach to rounding algorithms for the intersection of matroids.

\paragraph{Paper Organization.}
Section~\ref{sec:prelim} contains some required preliminaries.
\sgsp and its analysis appear in Section~\ref{sec:poisson}, where the fast approximation algorithms appear in Section~\ref{sec:applications}.


\section{Preliminaries}\label{sec:prelim}
In this section we briefly provide preliminaries regarding Poisson processes, matroids, and some properties of submodular functions. 

\paragraph{Poisson Processes.}
Let $ \{ N(t)\}_{t\geq \varepsilon}$ be a non-homogeneous Poisson process with rate function $ \lambda(t)$.
For convenience of presentation, we assume that the process starts at time $\varepsilon >0$ (for some given $\varepsilon$) and not $0$.
Thus, $ N(\varepsilon)=0$ and $N(t)$ denotes the number of events that occurred in the time interval $ [\varepsilon,t]$.
More generally, for every time interval $I$ we denote by $ N_I$ the number of events that occurred in interval $I$, e.g., if $ I=(t,t+\delta]$ then $ N_{(t,t+\delta]}=N(t+\delta)-N(t)$.

Our analysis requires the following properties of a Poisson process (see, e.g., \cite{ross2014introduction}, Chapter 5):
\begin{enumerate}
\item \label{Poisson-independ} If $ I$ and $J$ are disjoint intervals, then $ N_I$ and $ N_J$ are independent.
\item \label{Poisson-events} For every $ t\geq \varepsilon$: $\lim _{\delta \rightarrow 0^+}\frac{1}{\delta}\Pr[N_{(t,t+\delta]}=1]=\lambda(t)$ and $\lim _{\delta \rightarrow 0^+}\frac{1}{\delta}\Pr[N_{(t,t+\delta]}\geq 2]=0$.
\end{enumerate}

\paragraph{Matroids.} The following definition regarding matroid exchange maps is required that always exist (see Corollary 39.12a, \cite{schrijver2003combinatorial}).
\begin{definition}
\label{def:exchange}
    Given a matroid $\Mcal=(U,\Ical)$, a base $Z$ of $\Mcal$, and an independent set $A\in \Ical$, a matroid exchange map is a  function $h: Z\to A\cup\{\perp \}$ satisfying:
\begin{enumerate}
    \item For every $ j\in Z$: $A-h(j)+j\in \Ical$.
    \item For every $j\in Z\cap A$: $h(j)=j$.
    \item For every $i\in A$ there exists exactly one $j\in Z$ such that $h(j)=i$.
\end{enumerate}
\end{definition}
Recall that $\perp\notin U$ denotes a non-element, and thus $ A-\perp=A$ and \(\1_\perp = \vc{0}\).

\paragraph{Submodular Functions.}
We require the following properties of submodular functions (whose proofs appear in Appendix~\ref{app:properties}).


\begin{restatable}{lemma}{multilin}
\label{lem:multi-lin}
    Let \(\vc{a} \in [0,1]^n\) and fix \(k \in [n]\) such that \(a_k < 1\).
Let \(u, \ell \in [0,1]\) such that \(u \leq 1 - a_k\) and \(\ell \leq a_k\). Then
\[
    F(\vc{a} + u \mathbf{e}_k) - F(\vc{a} - \ell \mathbf{e}_k)
    =
    \frac{u + \ell}{1 - a_k} (F(\vc{a} \vee \mathbf{e}_k) - F(\vc{a})).
\]
\end{restatable}

\begin{restatable}{lemma}{lovaszestimate}
    \label{lem:lovasz-estimate}
    Let \(f\) be a non-negative submodular function, let \(p, q \in [0,1]\),
    let $A$ be a set and $R$ be a random set that includes  $\Pr[i\in R]\geq p$ for each $i\in A$ and $\Pr[i\in R]\leq q$  for each $i\not \in A$. Then 
    $$\E[f(R)]\geq (p-q) f(A).$$
\end{restatable}


\section{The \sgsplong}\label{sec:poisson}

Our goal in this section, is to prove \Cref{thrm:Poisson}.
Let \(\Mcal = (U, \Ical)\) be a matroid.
Let \(k\) denote the rank of the matroid.
We shall assume that we have at least \(k\) dummy elements that can be added to any independent set of size strictly less than $k$. Thus given any independent set we can extend it to a base by adding dummy elements.

\subsection{The Process}
\label{subsec:process}

We begin by defining a basic subroutine used in \sgsp, which we call a \emph{valid swap procedure}.
We use $\perp$ to denote an empty element, thus \(S - \perp = S\) for any set \(S\). 
\begin{definition}
\label{def:valid-swap-pair}
A \emph{swap procedure} is valid if,
given a set $A\in\Ical$ and a time $t\in(0,1]$, $\swap(t,A)$ returns a random pair 
$
(I,J)\in (A\cup\{~\perp~\})\times U
$
if it satisfies the following properties.

\begin{enumerate}[label=(\roman*)]
    \item\label{it-valid} With probability one, the update prescribed by $(I,J)$ is feasible. Namely, if \(I \neq J\), \(A - I + J \in \Ical\).

    \item\label{it-inter} If $J\in A$, then $I=J$.

    \item\label{it-i} For every $i\in A$,$
    \Pr[I=i]=\frac1k.$

    \item\label{it-j}  For  $j\in U$,
    $
    \Pr[J=j]\leq \frac1k.
    $
\end{enumerate}
\end{definition}
\noindent
Furthermore, we say a valid swap procedure is \emph{above-average} if, for every \(A \in \Ical\) and \(t \in (0,1]\),
\begin{equation}
\label{eq:above-avg}
    \E[ F(t\1_A\vee\1_{J})-F(t\1_A)]
    \geq
    \frac{1}{k}\left( F(t\1_A\vee\1_{\OPT})-F(t\1_A)\right).
\end{equation}
and \( \eta\)-\emph{almost-above-average}
if
\begin{equation}
\label{eq:almost-above-avg}
    \E[ F(t\1_A\vee\1_{J})-F(t\1_A)]
    \geq
    \frac{1}{k}\left( F(t\1_A\vee\1_{\OPT})-F(t\1_A)\right) - \frac{\eta}{k}.
\end{equation}
Let \(p_{ij}(t, A)\) denote the probability that the swap procedure will return the pair \((i,j)\).

\begin{algorithm}[!h]
\caption{\sgsp $(\Mcal,\varepsilon,\swap)$}
\SetKwInOut{Input}{input}
\SetKwInOut{Output}{output}

\SetAlgoNlRelativeSize{0}
\label{alg:Poisson}
 	
 	$ t\leftarrow \varepsilon$ and set $A \leftarrow \varnothing$.

    sample $\tau(t)$ the time of the next event after $t$ and $ t\leftarrow \tau(t)$.

    \While{$t<1$}{
 
    $(I, J)\leftarrow \swap(t,A)$. \label{poisson:item_sample}

    $A\leftarrow A-I+J$. \label{poisson:item_swap}

    \If{$I=J$}
    {with probability $t$, $A\leftarrow A-I$.\label{poisson:drop}} 

    sample $\tau(t)$ the time of the next event after $t$ and $ t\leftarrow \tau(t)$. 
    
    }
	
 	\Return $A$.
 	
\end{algorithm}
Given a valid \(\swap\) procedure we can formally define the \sgsp Algorithm. Our main goal in this section is to prove the following two lemmas that show the guarantee of the \sgsp under the presence of above-average swap procedure and $\eta$-almost above average swap procedure respectively. 

\begin{lemma}
\label{lem:average_poisson}
Given a matroid $ \Mcal=(U,\Ical)$ of rank $k$, a non-negative submodular function $f\colon 2^U \rightarrow \mathbb{R}_+$, a starting time $ 0<\varepsilon \leq 1$, and a right continuous above-average valid swap procedure $\swap$, Algorithm~\ref{alg:Poisson} finds $A\in \Ical$ satisfying $ \mathbb{E}[f(A)]\geq (1-\varepsilon)(\nicefrac{1}{e})\cdot f(\OPT) + e^{-1}f(\varnothing)$ in case \(f\) is non-monotone and $ \mathbb{E}[f(A)]\geq (1-\varepsilon)(1-\nicefrac{1}{e})\cdot f(\OPT) + e^{-1}f(\varnothing)$ in case \(f\) is monotone using in expectation $ k\ln{(1/\varepsilon)}$ swap procedure calls.
Here $\OPT=\argmax\{f(S) \colon S\in \Ical\}$ is an optimal base. 
\end{lemma}

\begin{lemma}
\label{lem:almost_average_poisson}
Given a matroid $ \Mcal=(U,\Ical)$ of rank $k$, a non-negative submodular function $f\colon 2^U \rightarrow \mathbb{R}_+$, a starting time $ 0<\varepsilon \leq 1$, and a right continuous  $\eta$-almost-above-average valid swap procedure $\swap$, Algorithm~\ref{alg:Poisson} finds $A\in \Ical$ satisfying $ \mathbb{E}[f(A)]\geq (1-\varepsilon)(\nicefrac{1}{e})\cdot f(\OPT) + e^{-1}\cdot f(\varnothing) - \eta$ in case \(f\) is non-monotone and $ \mathbb{E}[f(A)]\geq (1-\varepsilon)(1-\nicefrac{1}{e})\cdot f(\OPT) + e^{-1}f(\varnothing) - \eta$ in case \(f\) is monotone using in expectation $ k\ln{(1/\varepsilon)}$ swap procedure calls.
Here $\OPT=\argmax\{f(S) \colon S\in \Ical\}$ is an optimal base. 
\end{lemma}

\Cref{lem:average_poisson} with a right continuous above-average swap procedure implies \Cref{thrm:Poisson} that we provide in Section~\ref{sec:general}.

\subsection{Analysis of \sgsp}
Following the discussion in \Cref{sec:introductions}, 
our \sgsp process maintains an independent set
\(A(t)\) for every time \(\varepsilon \leq t \leq 1\).
As measured continuous greedy tracks the value of
\(F(x(t))\), we will look at the expected value of the indicator vector of \(A(t)\) scaled by \(t\). More formally,
define the potential
\[
Q(t) \coloneqq \E[F(t \1_{A(t)})].
\]
It will be useful for us to also look at this potential conditioned at the current state of the Poisson process.
Define
\[
V_{S, t}(r) = 
\E[F(r\1_{A(r)}) \; | \; A(t) = S]=
\sum_{T \in \Ical} F(r \1_T) \Pr[A(r) = T \;|\; A(t) = S].
\]

The analysis works by showing a rate of progress that must be satisfied by $V_{S,t}$ and thus by $Q(t)$. Lemma~\ref{lem:potential_average} states the progress when the valid swap procedure is above average. A crucial step in Measured Continuous Greedy is to bound the infinity norm of \(x(t)\).
In Lemma~\ref{lem:occupancy}, we do the analog analysis by bounding the probability an element belongs to the current independent set \(A(t)\). Finally, we solve the differential equation analyzing the progress of $Q(t)$ to obtain the proofs of Lemma~\ref{lem:average_poisson} and Lemma~\ref{lem:almost_average_poisson}. 
We begin by stating the following lemma about transition probabilities. Due to space restrictions, the proofs are deferred to Appendix~\ref{sec:appendix-proofs}. 
\begin{restatable}{lemma}{lemmatransition}
\label{lem:transit-prob}
Let \(t \in [\varepsilon, 1)\) and condition on the event \(A(t) = S\). If we run \Cref{alg:Poisson} with a valid swap procedure, it holds that
    \begin{align*}
    &\frac{\partial_+}{\partial r} \left(\Pr[A(r) = T \; | \; A(t) = S]\right)\big\rvert_{r = t}
    =\\
    &
    \begin{cases}
        \tfrac{k}{t}p_{ij}(t, S), & \exists i \in S \cup \{\perp\}, j \in U \setminus S, i \neq j : T = S - i + j\\
        kp_{ii}(t, S), & \exists i \in S: T = S - i\\
        \frac{-k}{t}
        \left(
        \sum_{i \in S \cup \{\perp\}}
        \sum_{j \in U: j \neq i} p_{ij}(t, S)
        +
        t\sum_{i \in S} p_{ii}(t, S)
        \right), & T = S\\
        0, &\text{otherwise}
    \end{cases}
    \end{align*}
\end{restatable}

We now characterize the rate of progress of $V_{S,t}$ and $Q_t$ for the above-average swap algorithm. Analysis for the $\eta$-almost above average swap algorithm is analogous and appears in the appendix.  The main challenge is to show that the differential equation below still holds despite the algorithm being spiteful and dropping elements. 
\begin{restatable}{lemma}{lempotaverage}
 \label{lem:potential_average}   
If \Cref{alg:Poisson} uses an above-average valid swap procedure, then for any \(t \in [\varepsilon, 1)\) and conditional on the event \(A(t) = S\).
    It holds that
    \[
    \frac{\partial_+}{\partial r} V_{S, t}(r) \big\rvert_{r = t} \geq F(t\1_S \vee \1_\OPT) - F(t\1_S).
    \]
    Therefore, it holds that
    \[
    Q'(t) \geq \E\!
    \left[
    F(\1_{\OPT}\vee t\1_{A(t)})
    \right]
    -
    Q(t).\]
 
\end{restatable}

Now we bound the probability an element belongs to the current independent set \(A(t)\). The proof crucially relies on the fact that items are dropped in Line~\ref{poisson:drop} of the algorithm.

\begin{restatable}{lemma}{occupancy}
\label{lem:occupancy}
    Let \(\Mcal = (U, \Ical)\) be a rank \(k\) matroid.
    Let \(A(t)\) be the independent set at time \(t\) in a run of \Cref{alg:Poisson} with a right continuous valid swap procedure.
    For every \(i \in U\), let
    \(p_i(t) \coloneqq \Pr(i\in A(t))\).
    Then, it holds that
    \[
    p_i(t) \leq \frac{1}{t} (1 - e^{-t + \varepsilon}).
    \]
\end{restatable}

Finally, we combine the results so far to solve a differential equation for the potential \(Q(t)\).

\begin{proof}[Proof of \Cref{lem:average_poisson}]

We start with the case when \(f\) is a monotone submodular function.
By \Cref{lem:potential_average}, we have
\[
Q'(r) \geq
\E\!\left[ F(\1_{\OPT}\vee r\1_{A(r)}) \right]
-
Q(r)
\geq
f(\OPT) - Q(r),
\]
where in the second inequality we used monotonicity. We now have
\[
\frac{\partial e^r Q(r)}{\partial r}
=
e^r Q(r) + e^rQ'(r)
\geq
e^r f(\OPT).
\]
By integrating from \(\varepsilon\) to \(t\), we obtain \(e^t Q(t) - e^\varepsilon Q(\varepsilon)
\geq (e^t - e^\varepsilon)f(\OPT)\).
Using the fact that \(Q(\varepsilon) = f(\varnothing)\), and evaluating the expression 
at \(t = 1\), we get
\[
Q(1) \geq \left(1 - \frac{e^\varepsilon}{e}\right)f(\OPT) + \frac{e^\varepsilon}{e} f(\varnothing)
\geq
(1 - \varepsilon)(1 - \nicefrac{1}{e}) f(\OPT)
+ e^{-1} f(\varnothing).
\]

Since \ref{it-valid} guarantees that \(A(t)\) is feasible with probability \(1\) for every \(t\), the set \(A(1)\) is a feasible solution with expected value \(Q(1)\).

We now look at the non-monotone case.
We again start with \Cref{lem:potential_average}:
\[
Q'(r) \geq \E\!\left[ F(\1_{\OPT}\vee r\1_{A(r)}) \right]
- Q(r).
\]
Instead of monotonicity,
we now apply \Cref{lem:lovasz-estimate}. Each element from \(\OPT\) appears with probability \(1\) and for each element not in \(\OPT\):
\[
\Pr[i \in R] = r \Pr[i \in A(r)] \leq 1 - e^{-r+\varepsilon}.
\]
Consequently,
\[
Q'(r) \geq e^{-r+\varepsilon} f(\OPT) -Q(r).
\]
We now have
\[
\frac{\partial e^r Q(r)}{\partial r}
=
e^r Q(r) + e^rQ'(r)
\geq
e^\varepsilon f(\OPT).
\]
By integrating from \(\varepsilon\) to \(t\), we obtain \(e^t Q(t) - e^\varepsilon Q(\varepsilon)
\geq (t-\varepsilon)e^\varepsilon f(\OPT)\).
Using the fact that \(Q(\varepsilon) = f(\varnothing)\), and evaluating the expression 
at \(t = 1\), we get
\[
Q(1) \geq (1-\varepsilon)\left(\frac{e^\varepsilon}{e}\right)f(\OPT) + \frac{e^\varepsilon}{e} f(\varnothing)
\geq
(1 - \varepsilon)(\nicefrac{1}{e}) f(\OPT)
+ e^{-1} f(\varnothing).
\]
Since \ref{it-valid} guarantees that \(A(t)\) is feasible with probability \(1\) for every \(t\), the set \(A(1)\) is a feasible solution with expected value \(Q(1)\).

The expected number of swap procedure calls is the expected number of events in the Poisson process:
\[
\int_\varepsilon^1 \lambda(t) dt
=
k
\int_\varepsilon^1 \frac{1}{t} dt
=
k \ln(1/\varepsilon). \qedhere
\]
\end{proof}
Proof of \Cref{lem:almost_average_poisson} follows analogously and appears in Appendix~\ref{sec:appendix-proofs}.

\section{Applications}\label{sec:applications}

\subsection{\texorpdfstring{Oracle to $F$}{}}\label{sec:general}
Now we prove Theorem~\ref{thrm:fastF} by giving appropriate swap procedures for a general matroid as well as for the special case of partition matroid. 
 \paragraph{General Matroid.}

In this section, we prove the following lemma.
\begin{lemma}
\label{lem:genF}
    Given a matroid \(\Mcal = (U, \Ical)\) of rank \(k\) and a non-negative submodular function, there is a polynomial time algorithm which is a right continuous above-average valid swap procedure and it makes
    \(O(n)\) calls to a multilinear oracle.
\end{lemma}

Note that \Cref{lem:genF} combined with \Cref{lem:almost_average_poisson} implies one part of \Cref{thrm:fastF}, since the expected number of calls to \(\swap\) is \(O(k \ln(1/\varepsilon))\) and each call makes \(O(n)\)
calls to a \(F(\cdot)\) oracle, resulting in a total of \(O(nk \ln(1/\varepsilon))\).

We start by giving the algorithm:
\begin{proof}
\begin{algorithm}[!h]
\caption{ $\swap(t,A)$ for General Matroids}
\SetKwInOut{Input}{input}
\SetKwInOut{Output}{output}

\SetAlgoNlRelativeSize{0}
\label{alg:swapmatroidF}
 	
 	Compute $ w_i \leftarrow F(t\1_A \vee \1_{i}) - F(t\1_A)$
    for every \(i \in U\)

    Compute \(Z \leftarrow \arg\max \{w(S): S \in \Ical\}\)
    \tcp{we may assume that \(Z\) is a base}

    Let $h\colon Z \rightarrow A\cup \{\perp \}$ be a matroid exchange map as in \Cref{def:exchange}.

    Sample random element \(j\) of \(Z\)
	
 	\Return $(h(j), j)$.
 	
\end{algorithm}
We will prove it is a valid swap procedure, and then that it is above-average.
Property \ref{it-valid} and \ref{it-inter} follow from the definition of \(h\).
Property \ref{it-i} holds since for any element \(i \in A\), for the event \(I = i\) to happen, it must be that
\(h^{-1}(i)\) is selected from \(Z\). By the properties of \(h\), there is only one element in the pre-image of \(i\), which is selected with probability \(1/k\).
Property \ref{it-j} follows because for every \(j \in U\),
for the event \(J = j\) to happen, first \(j \in Z\) and then
it must be the random element sampled, which happens with probability \(1/k\).

To prove it is above average, note that
\begin{align*}
    \E[F(t\1_A \vee \1_J) - F(t\1_A)]
    &=
    \frac1{k} \sum_{j \in Z} F(t\1_A \vee \1_j) - F(t\1_A)\\
    &\geq
     \frac1{k} \sum_{j \in \OPT} F(t\1_A \vee \1_j) - F(t\1_A) &&\text{by optimality of \(Z\)}\\   
     &\geq
     \frac{1}{k}(F(t\1_A \vee \OPT) - F(t\1_A))
     &&\text{by submodularity}
\end{align*}
where the last inequality is a standard application of submodularity~\cite{CCPV11}. The only calls to the oracle were to compute the weights in the first line, this can be done with \(n + 1\) calls.
\end{proof}

\begin{proof}[Proof of \Cref{thrm:Poisson}]
    \Cref{lem:genF} together with \Cref{lem:average_poisson} implies \Cref{thrm:Poisson}.
\end{proof}

 \paragraph{Partition Matroid.}

In this section, we prove the following lemma.
\begin{lemma}
\label{lem:partF}
    Given a partition matroid \(\Mcal = (U, \Ical)\) of rank \(k\) and a non-negative submodular function, there is a randomized polynomial time algorithm which is a right continuous above-average valid swap procedure and it makes, in expectation,
    \(O(n/k)\) calls to a multilinear oracle.
\end{lemma}

Note that \Cref{lem:partF} combined with \Cref{lem:almost_average_poisson} finalizes the proof of \Cref{thrm:fastF}, since the expected number of calls to \(\swap\) is \(O(k \ln(1/\varepsilon))\), each call makes in expectation \(O(n/k)\)
calls to a \(F(\cdot)\) oracle, and \(\swap\) is independent of the Poisson process, it results in a total of \(O(n \ln(1/\varepsilon))\).

\begin{proof}
Let \(P_1, \dots, P_k\) be the parts of \(\Mcal\) and thus
\(|A \cap P_j| \leq 1\) for any feasible set \(A\) and \(j \in [k]\).
Consider \Cref{alg:swappartF} .
We will prove it is a valid swap procedure, and then that it is above-average.

\begin{algorithm}[!h]
\caption{ $\swap(t,A)$ for Partition Matroid}
\SetKwInOut{Input}{input}
\SetKwInOut{Output}{output}

\SetAlgoNlRelativeSize{0}
\label{alg:swappartF}
 	
 	Sample random part \(j \in [k]\)
    
    Compute $ w_i \leftarrow F(t\1_A \vee \1_{i}) - F(t\1_A)$
    for every \(i \in P_j\)

    Set \(J\leftarrow \arg\max \{w_\ell: \ell \in P_j\}\)

    Set \(I \leftarrow \perp\) if \(A \cap P_j = \varnothing\), otherwise, let \(I\) be the unique element in \(A \cap P_j\)
	
 	\Return $(I, J)$.
 	
\end{algorithm}

For \ref{it-valid}, every part that is not \(j\) remains unchanged.
For part \(j\),
if \(A \cap P_j = \varnothing\), even after adding \(J\), the set \(A\) intersects \(P_j\) in at most one element and it is still feasible.
Otherwise,
if \(A \cap P_j\) is not empty, we will remove an element \(I\) from \(P_j\), which could be equal to \(J\), and so \(A - I + J\) will still intersect \(P_j\) in at most one element.
For \ref{it-inter}, if \(J \in A\), then \(J\) is the only element in \(A \cap P_j\) and thus \(I = J\).
For \ref{it-i}, for an element \(i \in A\) to be selected as \(I\), it is necessary and sufficient that the part it belongs to was selected, which happens with probability \(1/k\).
Similarly, for an element \(\ell \in U\) to be selected, it must happen at least that the part it belongs to was selected, thus \(\Pr[J = \ell] \leq 1/k\).

To prove it is above average,
define the weights for every part and
let \(Z\) be the set with the max-weight element from each part.
Note that
\begin{align*}
    \E[F(t\1_A \vee \1_J) - F(t\1_A)]
    &=
    \frac1{k} \sum_{j \in Z} F(t\1_A \vee \1_j) - F(t\1_A)\\
    &\geq
     \frac1{k} \sum_{j \in \OPT} F(t\1_A \vee \1_j) - F(t\1_A) &&\text{by the definition of \(Z\)}\\   
     &\geq
     \frac{1}{k}(F(t\1_A \vee \OPT) - F(t\1_A))
     &&\text{by submodularity}.
\end{align*}
where the last inequality is a standard application of submodularity~\cite{CCPV11}.
Note that when part \(j\) is sampled, we only need \(|P_j| + 1\) calls to \(F(\cdot)\). Thus, in expectation, the number of calls is \(\E[|P_j| + 1] = \tfrac{n}{k} + 1\).
\end{proof}
 \subsection{\texorpdfstring{Oracle to $f$}{}}

In this section, we prove Theorem~\ref{thrm:fastf} when we only have an oracle access to the submodular function $f$. The basic idea is to estimate $F$ using sampling. 
 On general instances, the number of samples required to attain good estimation is relatively high, and would penalize the running time. 
 To reduce the number of required samples, we provide $\eta$-almost-above-average swap algorithms, in which the values of $\eta$ depends on a  property of the instance named  {\em maximum sum of marginals} as well as on the value of the optimum. 
 We use the \sgsplong, together with these swaps algorithms, on instances which were first preprocessed, and show that the combination of the two leads to a fast algorithm with negligible loss to the approximation ratio.

 The {\em  Maximum Sum of Marginals} of the instance, defined as
 $$
 \mar(f,\Ical)= \max_{S\in \Ical} \sum_{i\in S} \left(f(\{i\})- f(\emptyset)\right).
 $$
 Our preprocessing procedure, described in the next lemma, is used to bound the value of the maximum sum of marginals. 
\begin{lemma}[Preprocessing]
\label{lema:imppre2}
There exists a polynomial time algorithm that 
given a non-negative submodular function $f:2^{U}\rightarrow \mathbb{R}_+$ and a matroid $\M=(U,\I)$ and $0<\delta \leq \frac12$, returns a (random) set $\bar{S}\in \I$ such that 
\begin{enumerate}
    \item $\max_{T\subseteq U: T\cup \bar{S}\in \I} \sum_{i\in T} f(\bar{S}\cup \{i\})-f(\bar{S})\leq \constwo \cdot \cons\max_{S\in \I} f(S).$ 
   \item $\E[\max_{T\subseteq U: T\cup \bar{S}\in \I}  f(\bar{S}\cup T)+  \frac{1}{2}\cdot f(\bar{S}) ]\geq (1-\delta) \max_{S\in \I} f(S).$  
\end{enumerate}
For a general matroid, the number of oracle calls made by the algorithm is $O(nk)$ and for a partition matroid, the  expected number of oracle  calls is $O\left(\frac{n}{\delta}\right)$.
\end{lemma}
 We give the proof of \Cref{lema:imppre2} in \Cref{sec:preprocessing}.
The properties of the swap algorithms we use are described in the following lemma. 
 \begin{lemma}
 \label{lem:f_swaps}
 For every $\delta>0$, 
there is an $\eta$-almost-above-average swap with $\eta= \delta \cdot \left( \mar(f,\Ical)+f(O)\right)$, where $O=\max_{S\in \Ical} f(S)$.  For general matroid the algorithm uses $O\left(\frac{n\cdot (k\cdot \ln n+\ln(1/\delta))}{\delta^2}\right)$  oracle calls to $f$. For partition matroid the swap algorithm uses  $O\left(\frac{n\cdot (k\cdot \ln n+\ln(1/\delta))}{k\cdot \delta^2}\right)$   oracle calls in expectation.
 \end{lemma}
We provide the swap algorithms for general matroids and partition matroids in \Cref{sec:generalf,sec:partitionf} respectively. We first show \Cref{thrm:fastf} can be derived using \Cref{lema:imppre2,lem:f_swaps}. 
\begin{proof}[Proof of \Cref{thrm:fastf}]
Given a Sub-Mat instance $(f,\Ical)$ we use the following simple algorithm
\begin{itemize}
\item Run the preprocessing algorithm from \Cref{lema:imppre2} on $(f,\Ical)$. Let $\bar{S}$ be the returned set.
\item Run the  \sgsp on $(g,\Ical/\bar{S})$, where $g(T)=f(T\cup \bar{S})$ with $\varepsilon= \frac{\varepsilon}{1000}$ and the swap process from \Cref{lem:f_swaps} with $\delta=\frac{\varepsilon}{1000}$, let $S$ be the returned set.
\item Return $S\cup \bar{S}$.
\end{itemize}

The preprocessing step requires $O(nk)$ queries for $f$. The execution of the \sgsp requires $O(k\ln \nicefrac{1}{\varepsilon})$ swap operations in expectation.
For (general) matroids each swap uses $O\left(\frac{n\cdot\left(k\cdot \ln n +\ln(1/\varepsilon)\right)}{\varepsilon^2}\right)$ function evaluations, which results in a total of   $O\left(k\cdot \ln(1/\varepsilon) \cdot  \frac{n\cdot\left(k\cdot \ln n +\ln(1/\varepsilon)\right)}{\varepsilon^2} \right) $ function evaluations in expectation. Similarly, for partition matroid each swap uses $O\left(\frac{n\cdot (k\cdot \ln n+\ln(1/\delta))}{k\cdot \delta^2}\right)$ oracle calls in expectation, which results in a total number of $O\left(\frac{n\cdot (k\cdot \ln n+\ln(1/\varepsilon))}{\varepsilon^2}\cdot \ln(1/\varepsilon)\right)$ oracle calls throughout the execution of the algorithm.

It can be easily observed that the function $g$ defined above is submodular, and that the returned solution is indeed feasible. 
Let $O=\argmax_{S\in \Ical} f(S)$ be an optimal solution for the original instance. 
By \Cref{lema:imppre2} it holds that $$
\mar(g,\Ical/\bar{S}) = \max_{T\in \Ical/\bar{S}} (g(\{i\})-g(\emptyset) ) = \max_{T\subseteq U : T\cup \bar{S}\in \Ical} \sum_{i\in T} (f(\bar{S}+i) -f(\bar{S}) )\leq 80 \cdot f(O). 
$$
In particular, this means that we run the \sgsp with an $\eta$-almost-above average swap algorithm in which $\eta = \delta \cdot (f(O)+ \mar(g,\Ical/\bar{S}) ) \leq \delta \cdot 90 \cdot f(O).$ 
Let $O_g= \argmax_{S\in \Ical/\bar{S}}g(S)$. By \Cref{lema:imppre2} it holds that $$\E\left[g(O_g)+\frac{1}{2}\cdot g(\emptyset)\right] = \E\left[f(\bar{S}\cup T)+\frac{1}{2} \cdot f(\bar{S})\right] = (1-\delta)\cdot f(O).$$
By \Cref{lem:almost_average_poisson}, this implies that $$
\begin{aligned}\E[g(S)|\bar{S}]&\geq \left(1-\frac{\varepsilon}{1000}\right) \cdot \nicefrac{1}{e}\cdot g(O_g)+e^{-1}\cdot g(\emptyset) -\eta \\
&\geq \left(1-\frac{\varepsilon}{1000}\right) \cdot \nicefrac{1}{e}\cdot \left(g(O_g)+ g(\emptyset) \right) -\delta \cdot 90\cdot f(O).
\end{aligned}$$
Therefore,
$$
\begin{aligned}
\E[f(S\cup \bar{S})] &= \E\left[ \E[g(S)|\bar{S}]\right] \\
&\geq \E\left[\left(1-\frac{\varepsilon}{1000}\right) \cdot \nicefrac{1}{e}\cdot \left(g(O_g)+ g(\emptyset) \right) -\delta \cdot 90\cdot f(O)\right] \\
&\geq \left(1-\frac{\varepsilon}{1000}\right) \cdot \frac{1}{e}\cdot (1-\delta)\cdot f(O) -\delta\cdot 90\cdot f(O) \\
&\geq (1-\varepsilon)\cdot e^{-1}\cdot f(O).
\end{aligned}
$$

\end{proof}

\subsubsection{General Matroids}

\label{sec:generalf}

We first present our swap procedure for general matroids. The procedure is the first half of the proof of \Cref{lem:f_swaps}. 
 \begin{algorithm}[!h]
\caption{Matroid Swap $f$}
\label{algo:mat_swap_f}
\vspace{4pt}\hrule\vspace{4pt}
\KwIn{$t\in[0,1]$, $A\in\mathcal I$}

Define $m=\frac{25}{2\cdot \delta^2} \cdot \left(k\ln (2n) +\ln 2+\ln (5/\delta) \right) $.

Sample $m$ sets $R_1,\ldots R_m \sim t\cdot \indicator_{A}$.

For every $i\in U$ compute $\tilde{w}_i = \frac{1}{m} \cdot \sum_{\ell=1}^{m}( f(R_{\ell}\cup \{i\}) -f(R_{\ell}))$.

Find a base $Z$ which maximizes $\sum_{i\in Z} \tilde{w}_i$.

Find a matroid exchange map $h:Z\to A\cup\{\perp\}$.

Sample $z\in Z$ uniformly at random. 

\Return{$I=h(z), J=z$}
\end{algorithm}

\begin{lemma}
\label{lem:mat_swap_f}
\Cref{algo:mat_swap_f} is an $\eta$-almost-above-average swap algorithm which uses $O\left(\frac{n\cdot (k\cdot \ln n+\ln(1/\delta))}{\delta^2}\right)$ oracle queries to $f$, where $\eta= \delta\cdot (f(O)+\mar(f,\Ical) )$ and $O=\argmax_{S\in\Ical} f(S)$.  
\end{lemma}

In the proof of \Cref{lem:mat_swap_f} we use the standard Hoeffding's inequality.
\begin{theorem}[Hoeffding's Inequality (e.g, Theorem 2.5 in \cite{Mc98})]
\label{thm:hoef}
Let $X_1,\ldots,X_n$ be independent random variables such that $a\leq X_k\leq b$ for all $k\in \{1,\ldots, n\}$. 
Define
$S=\sum_{k=1}^n X_k$
and let $\mu=\mathbb{E}[S]$.
Then, for every $t\ge 0$,
\[
\Pr\!\left(\left|S_n-\mu\right|\ge t(b-a)\right)
\le
2\exp\!\left(
-\frac{2t^2}{n}
\right).
\]
\end{theorem}
\begin{proof}[Proof of \Cref{lem:mat_swap_f}]
We first show the algorithm is indeed a swap algorithm. Since $h$ is a matroid exchange map it follows that $A-z+h(z)\in \Ical$ for all $z\in Z$, therefore $A-I+J\in \Ical$ always holds. By the same argument, for all $z\in Z\cap A$ we have $h(z)=z$, and therefore if $J\in A$ we also have $I=J$. Since $Z$ is a matroid base $|Z|=k$, and the probability of every item $z\in Z$ to be selected is exactly $\frac{1}{k}$. Every item $a\in A$ has exactly one item $z\in Z$ such that $a=h(z)$, hence the probability of $I=a$ is $\frac{1}{k}$. Finally, for every $j\in U$, $$\Pr(J=j)~~=~~ \begin{cases} \frac{1}{k} & j\in Z\\ 0& j\notin Z\end{cases} ~~\leq~~ \frac{1}{k}.$$
That is, we showed that \Cref{algo:mat_swap_f} is a swap algorithm.

The number of oracle calls the algorithm makes is $O(n\cdot m) = O\left(\frac{n\cdot (k\cdot \ln n+\ln(1/\delta))}{\delta^2}\right)$. 

For every $i\in U$ define $w_i = F(t\cdot \indicator_{A}\vee \indicator_{i}) -F(t\cdot \indicator_A)$. 
It can be easily observed that $w_i =\E[\tilde{w}_i] =\E[f(R_{\ell}+i)-f(R)]$ for all $\ell\in\{1,\ldots,m\}$.
In order to show that \Cref{algo:mat_swap_f} is $\eta$-almost-above average we first show that for every $Z\in \Ical$ the sum
$\sum_{i\in Z} \tilde{w}_i$ does not deviate afar from $\sum_{i\in Z} w_i$.  
\begin{claim}
\label{claim:Zbound}
For every $Z\in \Ical$ it holds that 
$$\Pr\left(\left|\sum_{i\in Z} \tilde{w}_i-\sum_{i\in Z} w_i \right|\geq \frac{\delta}{5}\cdot  (f(O)+\mar(f,\Ical)) \right)\,<\, 2\cdot \exp\left(-\frac{2}{25}\cdot\delta^2\cdot m\right).$$ 
\end{claim}
\begin{claimproof}
Define 
$X_{\ell} = \sum_{i\in Z} (f(R_{\ell}+ i)-f(R_{\ell}))$ for all $\ell \in \{1,\ldots, m\}$. It can be easily verified that 
$\E[X_{\ell}] = \sum_{i\in Z} w_i$ and that 
$\sum_{i\in Z} \tilde{w}_i = \frac{1}{m}\cdot \sum_{\ell=1}^m X_{\ell}$. 
Our goal is to apply Hoeffding’s inequality (\Cref{thm:hoef}). To do so, we first derive upper and lower bounds on the possible values of the random variables $X_\ell$. First, submodularity and by the definition of $\mar(f,\Ical)$, for all $\ell \in \{1,\ldots,m\}$ we have,
$$
X_{\ell} = \sum_{i\in Z} (f(R_{\ell}+i) -f(R_{\ell})) \leq \sum_{i\in Z'} (f(R_{\ell}+i)-f(R_{\ell}))\leq \mar(f,\Ical),
$$
where $Z'=\{i\in Z | f(R_{\ell}+i)-f(R_{\ell}) \geq 0 \}$. 
By submodularity we also have,
$$
X_{\ell} = \sum_{i\in Z} (f(R_{\ell} +i)-f(R_{\ell})) \geq f(R_{\ell}\cup Z) -f(R_{\ell}) \geq -f(R_{\ell}) \geq -f(O),
$$
where the last inequality holds as $R_{\ell}\subseteq A$, and therefore $R_{\ell}\in \Ical$. 
Overall, we showed that $-f(O)\leq X_{\ell} \leq \mar(f,\Ical)$ for all $\ell\in \{1,\ldots,m\}$. Therefore by  \Cref{thm:hoef} we have
$$
\begin{aligned}
    &\Pr\left(\left|\sum_{i\in Z} \tilde{w}_i\right|> \frac{\delta}{5}\cdot  (f(O)+\mar(f,\Ical)) \right) \\
    =& 
    \Pr\left( \left|\sum_{\ell=1}^{m} X_{\ell } - \E\left[\sum_{\ell=1}^{m} X_{\ell }\right ]\right| > m\cdot  \frac{\delta}{5} (f(O)+\mar(f,\Ical))   \right)\\
    \leq& 2\exp\left(-2\cdot \frac{\left(\frac{m\cdot \delta}{5}\right)^2}{m }\right)\\
    =&2\cdot\exp\left(-\frac{2}{25}\cdot \delta^2\cdot m \right)
\end{aligned}
$$

\end{claimproof}

Define an event $\Psi = \left\{ \forall Z\in \Ical:~~(\left|\sum_{i\in Z} \tilde{w}_i-\sum_{i\in Z} w_i \right|\geq \delta (f(O)+\mar(f,\Ical))\right\} $. By \Cref{claim:Zbound} and the union bound we have
$$
\Pr(\neg \Psi) \leq |\Ical|\cdot 2\cdot \exp\left(-\frac{2}{25}\cdot \delta^2\cdot m\right) \leq (2n)^k\cdot 2 \cdot \exp\left(-\frac{2}{25}\cdot \delta^2\cdot m\right)  \leq \exp(\ln(\delta /5)) =\frac{\delta}{5},
$$
which implies $\Pr(\Psi)\geq 1-\frac{\delta}{5}$. 

We evaluate the output of the algorithm for the case $\Psi$ has occurred or not separately.
\begin{itemize}
\item 
Subject to the assumption that $\Psi$ occured, then it holds that 
$$
\begin{aligned}
\sum_{i\in Z} w_i &\geq \sum_{i\in Z} \tilde{w}_i -\frac{\delta}{5} \cdot \left(\mar(f,\Ical) +f(O)    \right) \\
&\geq  \sum_{i\in O} \tilde{w}_i -\frac{\delta}{5} \cdot \left(\mar(f,\Ical) +f(O)    \right)\\
&\geq \sum_{i\in O} {w}_i -2
\cdot \frac{\delta}{5} \cdot \left(\mar(f,\Ical) +f(O)    \right),\\
\end{aligned}
$$
where the first and last inequalities hold as we assume $\Psi$ has happened, and the second inequality holds by the selection of $Z$ in the algorithm. 
The algorithm then returns each $z\in Z$ uniformly at random, therefore,
$$
\begin{aligned}
\E&[F(t\cdot \indicator_A\vee \indicator_{J} ) -F(t\cdot \indicator_A)| \Psi]\\ &= \E\left[
\frac{1}{k} \sum_{i\in Z }\left(F(t\cdot \indicator_A\vee \indicator_{\{i\}} ) -F(t\cdot \indicator_A)\right) \,\middle|\,\Psi\right] \\
&= \E\left[\frac{1}{k} \sum_{i\in Z }w_i |\Psi\right] \\
&\geq\frac{1}{k} \sum_{i\in O} {w}_i -\frac{1}{k}\cdot 2
\cdot \frac{\delta}{5} \cdot \left(\mar(f,\Ical) +f(O)    \right)\\
&= \frac{1}{k} \sum_{i\in O} \left( F(t\cdot \indicator_A\vee \indicator_{\{i\}} ) -F(t\cdot \indicator_A)\right) -\frac{2
\cdot \delta}{5\cdot k} \cdot \left(\mar(f,\Ical) +f(O)    \right) \\
&\geq \frac{1}{k}\cdot\left( F(t\cdot \indicator_A \vee \indicator_{O}) -f(t\cdot \indicator_A) \right) -\frac{2
\cdot \delta}{5\cdot k} \cdot \left(\mar(f,\Ical) +f(O)    \right). 
\end{aligned}
$$
\item  In case $\Psi$ does not occur we use a universal lower  bound over $\sum_{i\in S} F(t\cdot \indicator_{A}\vee \indicator_{\{i\}})- F(t\cdot \indicator_A)$ which holds for every independent set $S\in \Ical$. Specifically,  let $R\sim t\cdot \indicator_A$ be a random set. Then for all $S\in \Ical$ we have
$$
\begin{aligned}
\sum_{i\in S }F(t\cdot \indicator_{A}\vee \indicator_{\{i\}})- F(t\cdot \indicator_A)&= \sum_{i\in S }\E\left[f(R+i)-f(R) \right]\\
&= \E\left[\sum_{i\in S } (f(R+i)-f(R)) \right ]\\
&\geq \E\left[ f(R\cup S)-f(R) \right]\\
&\geq -\E[-f(R)] \\
&\geq -f(O),
\end{aligned}
$$
where the first inequality holds by submodularity, the second inequality holds as $f(R\cup S)\geq 0$ and the last inequality holds as $R\subseteq A\in \Ical$ and therefore $f(R)\leq f(O)$. By the above inequality we have 
$$
\begin{aligned}
\E&[F(t\cdot \indicator_A\vee \indicator_{J} ) -F(t\cdot \indicator_A)| \neg \Psi]\\ &= \E\left[
\frac{1}{k} \sum_{i\in Z }\left(F(t\cdot \indicator_A\vee \indicator_{\{i\}} ) -F(t\cdot \indicator_A)\right) \,\middle|\,\neg \Psi\right] \\
&\geq \E\left[-\frac{1}{k} f(O)\,\middle|\,\neg \Psi\right] \\
&\geq -\frac{1}{k}\cdot f(O).
\end{aligned}
$$
\end{itemize}
Using the above two cases together we get
$$
\begin{aligned}
\E&\left[F(t\cdot \indicator_A\vee \indicator_{J} ) -F(t\cdot \indicator_A)\right]\\
&= \Pr(\Psi)\cdot 
\E[F(t\cdot \indicator_A\vee \indicator_{J} ) -F(t\cdot \indicator_A)| \Psi]+\Pr(\neg \Psi)\cdot 
\E[F(t\cdot \indicator_A\vee \indicator_{J} ) -F(t\cdot \indicator_A)| \neg \Psi]\\
&\geq \Pr(\Psi)\cdot \left( \frac{1}{k}\cdot\left( F(t\cdot \indicator_A \vee \indicator_{O}) -F(t\cdot \indicator_A) \right) -\frac{2
\cdot \delta}{5\cdot k} \cdot \left(\mar(f,\Ical) +f(O)    \right) \right) \\
&~~~~~+ \Pr(\neg \Psi)\cdot \left(-\frac{1}{k}\cdot f(O).\right)\\
&\geq \Pr(\Psi)\cdot \frac{1}{k} \cdot F(t\cdot \indicator_{A}\vee \indicator_O) -\frac{1}{k}\cdot F(t\cdot \indicator_A) -\frac{3\cdot \delta }{5\cdot k}\cdot (\mar(f,\Ical) + f(O))\\
&\geq \left(1-\frac{\delta}{5}\right) \cdot \frac{1}{k} \cdot F(t\cdot \indicator_{A}\vee \indicator_O) -\frac{1}{k}\cdot F(t\cdot \indicator_A) -\frac{3\cdot \delta }{k}\cdot (\mar(f,\Ical) + f(O))\\
&\geq \frac{1}{k} \cdot F(t\cdot \indicator_{A}\vee \indicator_O) -\frac{1}{k}\cdot F(t\cdot \indicator_A) -\frac{5\cdot \delta }{5\cdot k}\cdot (\mar(f,\Ical) + f(O))
\end{aligned}
$$
where  the second inequality holds as $\Pr(\neg \Psi)\leq \frac{\delta}{5}$ and $\Pr(\Psi)\leq 1$, and the last inequality holds as $\Pr(\Psi)\geq 1-\delta/5 $ and $F(t\cdot\indicator_A\vee \indicator_O)\leq 2\cdot f(O)$. Overall, we showed that \Cref{algo:mat_swap_f} is $\eta$-almost-above average, as required.

\end{proof}

\subsubsection{Partition Matroid}
\label{sec:partitionf}

We present a faster implementation of \Cref{algo:mat_swap_f} for partition matroids.
Recall that in a partition matroid the ground set $U$ is partitioned into $k$ parts
$P_1,\ldots,P_k$, and a set $S\subseteq U$ is independent if and only if
$|S\cap P_j|\leq 1$ for every $j\in\{1,\ldots, k\}$.

Consider the logic of \Cref{algo:mat_swap_f}. It first computes  estimates $\tilde{w}_i$ of $F(t\cdot \indicator_{A}\vee \indicator_{\{i\}})-F(t\cdot \indicator_{A})$ for all $i\in U$. The  algorithm then uses these estimates to find a  base $Z$  maximizing $\sum_{i\in Z} \tilde{w}_i$.    Finally, the algorithm returns a random item $z\in Z $ together with an item in $A$ that can be swapped with it.

For partition matroids we can reconstruct the same logic while only estimating $\tilde{w}_i$ for items in a single part $P_j$. As in \Cref{algo:mat_swap_f}, the algorithm begins by sampling random sets $R_1,\ldots, R_m\sim t\cdot \indicator_{A}$. 
For every $i\in U$ let  $\tilde{w}_i = \frac{1}{m}\sum_{\ell=1}^{m} (f(R_{\ell}+i) -f(R_{\ell}))$ denote the estimate of $F(t\cdot \indicator_{A} \vee \indicator_{\{i\}}) -F(t\cdot  \indicator_A)$
based on the samples $R_1,\ldots, R_m$. Although these estimates are well defined for every $i\in U$, the algorithm does not compute all of them.
\Cref{algo:mat_swap_f} finds a base $Z$ maximizing $\sum_{i\in Z}\tilde{w}_i$, computes a matroid exchange map $h$ from $Z$ to $A$, picks $z\in Z$ uniformly at random, and returns $J=z$ and $I=h(z)$. The analysis of \Cref{algo:mat_swap_f} (\Cref{lem:mat_swap_f}) implies that 
$$\E[ F(t\cdot \indicator_{A}\vee \indicator_{\{J\}} ) -F(t\cdot \indicator_A)] \geq \frac{1}{k}\cdot\left( F(t\cdot \indicator_A \vee \indicator_{O}) -F(t\cdot \indicator_A) \right) -\frac{
\delta}{ k} \cdot \left(\mar(f,\Ical) +f(O)    \right).$$

In the case of partition matroids, $Z=\{z_1,\ldots, z_k\}$ with $z_j =\argmax_{i\in P_j}  \tilde{w}_i$. Then, sampling a uniformly random element of $Z$ is equivalent to sampling $j\in [k]$ uniformly at random, and selecting $z_j$. Moreover, the exchange map $h$ must satisfy  $h(z_j)=\begin{cases} a_j &A\cap P_j = \{a_j\}\\
\perp & A\cap P_j=\emptyset\end{cases}$ for all $j\in \{1,\ldots,k\}$. Therefore, to simulate the execution of \Cref{algo:mat_swap_f}, we can simply sample $j\in \{1,\ldots, k\}$ uniformly at random, find $z_j = \argmax_{i\in P_j} \tilde{w}_i$, and return $z_j $ together with $a_j$  if $P_j\cap A=\{a_j\}$ or $\perp$ if $P_j\cap A=\emptyset$.

The resulting algorithm is given in
\Cref{algo:part_swap_f}. As the algorithm only computes $\tilde{w}_i$ for $i\in P_j$, the number  of oracle queries it uses is $O(|P_j|\cdot m)$. As  $\E[|P_j|]=\frac{n}{k}$, we get that the number of oracle calls in expectation is
$$
O\left(\frac{n}{k}\cdot m \right) =O\left(\frac{n}{\delta^2}\cdot \left(\ln n +\frac{\ln(1/\delta)}{k}\right)\right).
$$
The following lemma follows from the above discussion.
\begin{lemma}
Algorithm \Cref{algo:part_swap_f} is an $\eta$-almost-above-average swap algorithm for partition matroids  which uses $O\left(\frac{n}{\delta^2}\cdot \left(\ln n +\frac{\ln(1/\delta)}{k}\right)\right)$ oracle queries to $f$ in expectation , where $\eta= \delta\cdot (f(O)+\mar(f,\Ical) )$ and $O=\argmax_{S\in\Ical} f(S)$.  
\end{lemma}

 \begin{algorithm}[!t]
\caption{Partition Swap $f$}
\label{algo:part_swap_f}
\vspace{4pt}\hrule\vspace{4pt}
\KwIn{$t\in[0,1]$, $A\in\mathcal I$}

Define $m=\frac{25}{2\cdot \delta^2} \cdot \left(k\ln (2n) +\ln 2+\ln (5/\delta) \right) $.

Sample $m$ sets $R_1,\ldots R_m \sim t\cdot \indicator_{A}$.

Sample $j\in \{1,\ldots ,k\}$ uniformly at random. 

For every $i\in P_j$ compute $\tilde{w}_i = \frac{1}{m} \cdot \sum_{\ell=1}^{m}( f(R_{\ell}\cup \{i\}) -f(R_{\ell}))$.

Let $z_j = \argmax_{i\in P_j} \tilde{w}_i$. 

Define $a_j = \begin{cases} i &A\cap P_j = \{i\}\\
\perp & A\cap P_j=\emptyset\end{cases}$

\Return{$I=a_j, J=z_j$}
\end{algorithm}

\newpage

\appendix

\section{Properties of Submodular Functions}\label{app:properties}

\multilin*

\begin{proof}
If \(u = \ell = 0\), both sides are equal to \(0\), otherwise,
    since \(F\) is linear in the \(k\)-th coordinate, it holds that
    \[
    \frac{F(\vc{a} + u \mathbf{e}_k) - F(\vc{a} - \ell \mathbf{e}_k)}{u + \ell}
    =
    \nabla_k F(\vc{a}) = \frac{F(\vc{a} \vee \mathbf{e}_k) - F(\vc{a})}{1 - a_k}.
    \]
    Rearranging the terms leads to the desired equation.
\end{proof}

\lovaszestimate*
\begin{proof}
     Let $\mathbf{\alpha}\in [0,1]^{U}$ denote the probability marginal vector of the random set \(R\), i.e. $Pr[i\in R]=\alpha_i$.
     Note that \(\alpha_i \geq p\) for every \(i \in A\) and \(\alpha_i \leq q\) for every \(i \not \in A\).
     Since the Lovász extension is the convex closure of a submodular set function, we have
     \[
     \E[f(R)] \geq F_L(\alpha).
     \]
 But then for the Lovasz extension $F_L(\alpha)$, we have
\begin{align*}
    F_L(\alpha):=E_{r\in  \mathrm{unif}[0,1]}[f(\{i:\alpha_i\geq r\})]
\end{align*}
But when $q< r< p$, the threshold set $\{i:\alpha_i\geq r\} = A$. Thus the above expectation is at least $(p-q)f(A)$ where we use non-negativity of $f$.
If \(p \leq q\), the inequality holds trivially by non-negativity.
\end{proof}

\section{\texorpdfstring{Proofs from \Cref{sec:poisson}}{}}
\label{sec:appendix-proofs}

\lemmatransition*

Before proving \Cref{lem:transit-prob}, we will prove a useful lemma relating the probability derivative with the Poisson Process.
The proof is the same as in Claim 3.6 in \cite{GanzKSS26a}.

\begin{restatable}{claim}{cltransition}
For every time \(\varepsilon \leq t < 1\) and \(S, T \in \Ical\):
\label{cla:partial}
    \[
    \frac{\partial_+}{\partial r} \left(\Pr[A(r) = T \; | \; A(t) = S]\right)\big\rvert_{r = t}
    =
    \lambda(t)
    \left(
    \lim_{\delta \to 0^+} \Pr[A(t + \delta) = T \; | \; A(t) = S, N = 1] - \1_{T = S}
    \right)
    \]
\end{restatable}

\begin{proof}
First, we note that the definition of a right derivative gives that:
\begin{align}
&\frac{\partial_+ \Pr [A(r)=T|A(t)=S]}{\partial r}\bigg|_{r=t}  = \nonumber\\&\lim_{\delta\rightarrow 0^+}\frac{1}{\delta} \Big( \Pr\left[A(t+\delta)=T|A(t)=S\right] - \Pr[A(t)=T|A(t)=S]\Big) =\nonumber \\
&  \lim _{\delta \rightarrow 0^+}\frac{1}{\delta}\Big( \Pr\left[A(t+\delta)=T|A(t)=S\right] - 1_{S=T}\Big)
\label{eq1-dervidef}
\end{align}
where \eqref{eq1-dervidef} follows from the observation that $\Pr[A(t)=T|A(t)=S]$ equals $1$ if $S=T$ and $0$ otherwise.

The law of total probability applied to $ \Pr[A(t+\delta)=T|A(t)=S]$ and conditioned on the number of events in the interval $ (t,t+\delta]$ being either $0$, $1$, or at least $2$, yields that:
\begin{align}
\Pr&[A(t+\delta)=T|A(t)=S]  \nonumber =\\ & \Pr[N=0|A(t)=S]\cdot \Pr[A(t+\delta)=T|A(t)=S,N=0] + \nonumber \\
& \Pr[N=1|A(t)=S]\cdot \Pr[A(t+\delta)=T|A(t)=S,N=1] + \nonumber \\
& \Pr[N\geq 2|A(t)=S]\cdot \Pr[A(t+\delta)=T|A(t)=S,N\geq 2] . \label{eq2-totalprob}
\end{align}
It is important to note that the random variable $A(t)$ depends only on all the random choices in the time interval $ [\varepsilon,t]$, i.e., the randomness of the Poisson process and the randomness of elements selection (step \ref{poisson:item_sample} in Algorithm \ref{alg:Poisson}) in the time interval $ [\varepsilon,t]$.
In contrast, the random variable $ N$ depends only on the randomness of the Poisson process in the time interval $ (t,t+\delta]$.
Since these two time intervals are disjoint, property \ref{Poisson-independ} of a Poisson process implies that $A(t)$ and $ N$ are independent random variables.
Hence, we can conclude from \eqref{eq2-totalprob} above that:
\begin{align}
(\Pr&[A(t+\delta)=T|A(t)=S]-1_{T=S}) =\nonumber\\ &
\Pr[N=0]\cdot \left(\Pr[A(t+\delta)=T|A(t)=S,N=0]-1_{T=S} \right) + \nonumber \\
& \Pr[N=1]\cdot \left(\Pr[A(t+\delta)=T|A(t)=S,N=1]-1_{T=S} \right) + \nonumber \\
& \Pr[N\geq 2]\cdot \left(\Pr[A(t+\delta)=T|A(t)=S,N\geq 2]-1_{T=S} \right)    =\nonumber  \\
& \Pr[N=1]\cdot \left(\Pr[A(t+\delta)=T|A(t)=S,N=1]-1_{T=S} \right) + \nonumber \\
& \Pr[N\geq 2]\cdot \left(\Pr[A(t+\delta)=T|A(t)=S,N\geq 2]-1_{T=S} \right) . \label{eq3new}
\end{align}
Equality \eqref{eq3new} follows from the simple observation that $ \Pr[A(t+\delta)=T|A(t)=S,N=0]$ equals $1$ if $T=S$ and $0$ otherwise. 

Plugging \eqref{eq3new} into \eqref{eq1-dervidef} gives:
\begin{align}
&\frac{\partial_+ \Pr [A(r)=T|A(t)=S]}{\partial r}\bigg|_{r=t}\nonumber =\\ & \lim _{\delta \rightarrow 0^+} \frac{1}{\delta}\Pr [N=1]\left( \Pr[A(t+\delta)=T|A(t)=S,N=1]-1_{T=S}\right) + \nonumber \\
& \lim _{\delta \rightarrow 0^+} \frac{1}{\delta}\Pr[N\geq 2] \left( \Pr[A(t+\delta)=T|A(t)=S,N\geq 2] - 1_{T=S}\right)  =\nonumber \\
 & \lambda(t) \left( \lim_{\delta\to 0^+}\Pr[A(t+\delta)=T|A(t)=S,N=1] - 1_{T=S}\right). \label{eq4new}
\end{align}

Equality \eqref{eq4new} follows from property \ref{Poisson-events} of a Poisson process, i.e., $\lim _{\delta \rightarrow 0^+}\frac{1}{\delta}\Pr[N_{(t,t+\delta]}=1]=\lambda(t)$ and $\lim _{\delta \rightarrow 0^+}\frac{1}{\delta}\Pr[N_{(t,t+\delta]}\geq 2]=0$.
\end{proof}

\begin{proof}[Proof of \Cref{lem:transit-prob}]
    There are two ways a set can change in one event.
    Either \(I \neq J\), and a swap operation is executed, and note that since \(I\) can be \(\perp\), this swap may be simply the addition of element \(J\).
    Or \(I = J\), in which case the element may be dropped.
    The set can also stay unchanged.
    All other sets cannot be reached in a single event.
    Using \Cref{cla:partial}, we will prove the correctness of the statement case by case.

    For the first case, we have by definition,
    \[
    \lim_{\delta \to 0^+} \Pr[A(t + \delta) = S - i + j \; | \; A(t) = S, N = 1] = p_{ij}(t, S),
    \]
    and \(\lambda(t) = k/t\). By \Cref{cla:partial}, we are done.
    
    For the second case,
    we need \(I = J = i\), and \(i\) to be dropped with probability \(t\).
    This is the only way an element can be dropped since if \(J\) already belongs to \(A\), then \(I = J\) by \ref{it-inter} in \Cref{def:valid-swap-pair} .
    Thus,
    \[
    \lim_{\delta \to 0^+} \Pr[A(t + \delta) = S - i \; | \; A(t) = S, N = 1] = t p_{ii}(t, S).
    \]
    Again, \Cref{cla:partial} finishes the proof of this case.
    
    For the third case, it must be that \(I = J\) and the element was not dropped.
    The case \(I = \perp\) and \(J \in A\) cannot happen by \ref{it-inter} in \Cref{def:valid-swap-pair}.
    Thus,
    \[
    \lim_{\delta \to 0^+} \Pr[A(t + \delta) = S \; | \; A(t) = S, N = 1] = (1 - t)\sum_{i \in S} p_{ii}(t, S).
    \]
    By \ref{it-i} in \Cref{def:valid-swap-pair},
    for every \(i \in S\), it holds that 
    \[
    p_{ii}(t, S) = \frac{1}{k} - \sum_{j \in U: j\neq i} p_{ij}(t, S),
    \]
    and
    \[
    \sum_{j \in U} p_{\perp j}(t, S)
    =
    1 - \sum_{i \in S} \sum_{j \in U} p_{ij}(t, S)
    =
    1 - \frac{|S|}{k}.
    \]
    
    Thus, combining \Cref{cla:partial} with the previous observations
    \begin{align*}
    \frac{\partial_+}{\partial r} \left(\Pr[A(r) = S \; | \; A(t) = S]\right)\big\rvert_{r = t} 
    &=
    \frac{k}{t}\left((1 - t) \sum_{i \in S} p_{ii}(t, S) - 1\right)\\
    &=
    \frac{-k}{t}
    \left(
    1 - \sum_{i \in S} p_{ii}(t, S) + t\sum_{i \in S} p_{ii}(t, S)
    \right)\\
    &=
    \frac{-k}{t}
    \left(
    1 - \frac{|S|}{k} + \sum_{i \in S} \sum_{j \in U: j \neq i} p_{ij}(t, S) + t\sum_{i \in S} p_{ii}(t, S)
    \right)\\
    &=
    \frac{-k}{t}
    \left(
    \sum_{i \in S \cup \{\perp\}} \sum_{j \in U: j \neq i} p_{ij}(t, S) + t\sum_{i \in S} p_{ii}(t, S)
    \right).
    \end{align*}
    
    All other sets cannot be reached in a single event. Thus,
    \[
    \lim_{\delta \to 0^+} \Pr[A(t + \delta) = T \; | \; A(t) = S, N = 1] = 0. \qedhere
    \]
\end{proof}

\lempotaverage*

Before proving \Cref{lem:potential_average}, we prove a Lemma bounding the growth of of \(V_{S, t}(r)\). After that the Lemma will follow by applying the above-average property of our swap procedure.

\begin{lemma}
\label{lem:cond-pot}
Let \(t \in [\varepsilon, 1)\) and condition on the event \(A(t) = S\).
 If we run \Cref{alg:Poisson} with a valid swap procedure, it holds that
    \[
    \frac{\partial_+}{\partial r} V_{S, t}(r) \big\rvert_{r = t} \geq k\E[F(t\1_S \vee \1_J) - F(t\1_S)].
    \]
\end{lemma}

\begin{proof}
   We have
    \begin{align*}
        \frac{\partial_+}{\partial r} V_{S, t}(r)\big\vert_{r = t}
        &=
        \sum_{T \in \Ical}\frac{\partial_+}{\partial r} (F(r \1_T) \Pr[A(r) = T
        \;|\; A(t) = S])\big\rvert_{r = t}\\
        &=
        \sum_{T \in \Ical} \frac{\partial_+}{\partial r} (F(r \1_T))\big\rvert_{r = t}
        \Pr[A(t) = T \;|\; A(t) = S]\\
        &+
        \sum_{T \in \Ical} F(t \1_T)
        \frac{\partial_+}{\partial r}(\Pr[A(r) = T\;|\; A(t) = S])\big\rvert_{r = t}.
    \end{align*}
By Chain rule and \Cref{lem:transit-prob}, we get
\begin{align*}
    \frac{\partial_+}{\partial r} V_{S, t}(r)\big\vert_{r = t}
    &=
    \sum_{i \in S} \nabla_i F(t \1_S)
    +
    \frac{k}{t}\sum_{i \in S \cup \{\perp\}}\sum_{j \in U: j \neq i} F(t\1_{S - i + j})p_{ij}(t, S)\\
    &+
    k \sum_{i \in S} F(t\1_{S-i})p_{ii}(t, S)
    -
    \frac{k F(t\1_S)}{t}
        \left(
        \sum_{i \in S \cup \{\perp\}}
        \sum_{j \in U: j \neq i} p_{ij}(t, S)
        +
        t\sum_{i \in S} p_{ii}(t, S)
        \right)\\
    &=
    \frac{k}{t}\sum_{i \in S \cup \{\perp\}}\sum_{j \in U: j \neq i}p_{ij}(t, S)
    (F(t\1_{S - i + j}) - F(t\1_S))\\
    &+
    \sum_{i \in S}
    \frac{F(t\1_S) - F(t\1_{S-i})}{t}
    - kp_{ii}(t, S) (F(t\1_S) - F(t\1_{S-i})),
\end{align*}
where in the first equality we used \(p_{ij}(t, S) = 0\)
for every \(j \in S\) such that \(j \neq i\) by \ref{it-inter}
and in the second equality we used the multilinearity of \(F\) to expand the gradient and reorganized the terms.

We start analyzing the first sum.
Note that
\begin{align*}
F(t\1_{S - i + j}) - F(t\1_S)
&=
F(t\1_{S - i + j}) - F(t\1_{S - i}) - (F(t\1_S) -
F(t\1_{S-i}))\\
&\geq
F(t\1_{S + j}) - F(t\1_{S}) - (F(t\1_S) -
F(t\1_{S-i})),
\end{align*}
where in the inequality we used standard diminishing returns property of the multilinear extension \cite{CCPV11}.
Thus, it is lower bounded by
\[
\frac{k}{t}
\left(
\sum_{i \in S \cup \{\perp\}}\sum_{j \in U: j \neq i}p_{ij}(t, S) (F(t\1_{S + j}) - F(t\1_S))
-
\sum_{i \in S \cup \{\perp\}}\sum_{j \in U: j \neq i}p_{ij}(t, S) (F(t\1_{S}) - F(t\1_{S-i}))
\right)
\]
The first double sum can be written as
\begin{equation}
\label{eq:step-one}
    \sum_{i \in S \cup \{\perp\}}\sum_{j \in U: j \neq i}p_{ij}(t, S) (F(t\1_{S + j}) - F(t\1_S))
    =
    t \sum_{i \in S \cup \{\perp\}}\sum_{j \in U: j \neq i}p_{ij}(t, S) (F(t\1_{S} \vee \1_j) - F(t\1_S)).
\end{equation}
We used the fact that if \(j \in S\), then \(p_{ij}(t, S) = 0\) for every \(j \neq i\) by \ref{it-inter} in \Cref{def:valid-swap-pair}, thus we only need to look at terms with \(j \not\in S\).
For those, we applied \Cref{lem:multi-lin} with
\(\mathbf{a} = t\1_S\), \(k = j\), \(u = t\), and \(\ell= 0\).

The second term can be simplified
\begin{align*}
    \sum_{i \in S \cup \{\perp\}}\sum_{j \in U: j \neq i}p_{ij}(t, S) (F(t\1_{S}) - F(t\1_{S-i}))
    &=
    \sum_{i \in S \cup \{\perp\}}
    (F(t\1_{S}) - F(t\1_{S-i}))
    \sum_{j \in U: j \neq i}p_{ij}(t, S)\\
    &=
    \sum_{i \in S}
    (F(t\1_{S}) - F(t\1_{S-i}))
    \left(\frac1k - p_{ii}(t, S)\right),
\end{align*}
where in the second equality we used \ref{it-i} from \Cref{def:valid-swap-pair}.
Combining with \eqref{eq:step-one}, we get:
\begin{align*}
    \frac{k}{t}\sum_{i \in S \cup \{\perp\}}\sum_{j \in U: j \neq i}p_{ij}(t, S)
    (F(t\1_{S - i + j}) - F(t\1_S))
    &\geq
    k \sum_{i \in S \cup \{\perp\}}\sum_{j \in U: j \neq i}p_{ij}(t, S) (F(t\1_{S} \vee \1_j) - F(t\1_S))\\
    &\hspace{-20mm}-
    \sum_{i \in S} \left(\frac{F(t\1_S) - F(t\1_{S-i})}{t} - \frac{k}{t}p_{ii}(t, S) (F(t\1_S) - F(t\1_{S-i}))\right).
\end{align*}

Returning to the original expression, we have that the terms \(\frac{F(t\1_S) - F(t\1_{S-i})}{t}\) will cancel, and we are left with
\begin{align*}
\frac{\partial_+}{\partial r} V_{S, t}(r)\big\vert_{r = t}
&\geq
k \sum_{i \in S \cup \{\perp\}}\sum_{j \in U: j \neq i}p_{ij}(t, S) (F(t\1_{S} \vee \1_j) - F(t\1_S))\\
&+
k\sum_{i \in S}p_{ii}(t, S)\left(\frac1{t} - 1\right) (F(t\1_S) - F(t\1_{S-i}))
\end{align*}
Observe that for every \(i \in S\), it holds that
\[
F(t\1_S) - F(t\1_{S - i}) = \frac{t}{1 - t} (F(t\1_S \vee \1_i) - F(t\1_S))
\]
by applying \Cref{lem:multi-lin} with \(\mathbf{a} = t\1_S\), \(k = i\), \(u = 0\), and \(\ell = t\).
Thus,
\begin{align*}
\frac{\partial_+}{\partial r} V_{S, t}(r)\big\vert_{r = t}
&\geq
k \sum_{i \in S \cup \{\perp\}}\sum_{j \in U: j \neq i}p_{ij}(t, S) (F(t\1_{S} \vee \1_j) - F(t\1_S))\\
&+
k\sum_{i \in S}
p_{ii}(t, S) (F(t\1_S \vee \1_i) - F(t\1_{S}))\\
&=
k \sum_{i \in S \cup \{\perp\}}\sum_{j \in U}p_{ij}(t, S) (F(t\1_{S} \vee \1_j) - F(t\1_S))\\
&=
k\E[F(t\1_{S} \vee \1_J) - F(t\1_S)].
\end{align*}
In the first equality, we used the fact that for
\(i = \perp\) the term \(F(t\1_S \vee \1_i) - F(t\1_{S})\) is \(0\) and completed the missing term \(j = i\) in the first sum.
\end{proof}

\begin{proof}[Proof of \Cref{lem:potential_average}]
    From \Cref{lem:cond-pot} and \eqref{eq:above-avg},
    we immediately obtain
    \begin{equation}
    \label{eq:partial-bound}
    \frac{\partial_+}{\partial r} V_{S, t}(r) \big\rvert_{r = t} \geq F(t\1_S \vee \1_\OPT) - F(t\1_S).
    \end{equation}
    Furthermore,
    By the definition of \(Q(r)\) and \(V_{S, t}(r)\), for every $r \geq t$,
\[
Q(r)
=
\sum_{S\in\Ical}
\Pr(A(t)=S)\cdot V_{S,t}(r).
\]
Hence,
\begin{align*}
Q'(t)
&=
\sum_{S\in\Ical}
\Pr(A(t)=S)
\left.
\frac{\partial V_{S,t}(r)}
{\partial r}
\right|_{r=t}
\\
&\ge
\sum_{S\in\Ical}
\Pr(A(t)=S)
\left(
F(\1_{\OPT}\vee t\1_S)
-
F(t\1_S)
\right)
&&\text{by \eqref{eq:partial-bound}}
\\
&=
\sum_{S\in\Ical}
\Pr(A(t)=S)
F(\1_{\OPT}\vee t\1_S)
-
\sum_{S\in\Ical}
\Pr(A(t)=S)
F(t\1_S)
\\
&=
\E\!\left[
F(\1_{\OPT}\vee t\1_{A(t)})
\right]
-
Q(t).
\end{align*}
    
\end{proof}

\occupancy*

\begin{proof}

We start by analyzing the probability an element not in the current set is added to it.
By definition,
\begin{align*}
&\left.
\frac{\partial}{\partial r}
\Pr(i\in A(r)\mid i\notin A(t))
\right|_{r=t}
=
\lim_{\delta\to0^+}
\frac1\delta
\Pr(i\in A(t+\delta)\mid i\notin A(t)).
\end{align*}
Let $N$ denote the number of Poisson events in $(t,t+\delta]$.
Then
\begin{align*}
\Pr(i\in A(t+\delta)\mid i\notin A(t))
&=
\Pr(N=0)\Pr(i\in A(t+\delta)\mid i\notin A(t),N=0)
\\
&+
\Pr(N=1)\Pr(i\in A(t+\delta)\mid i\notin A(t),N=1)
\\
&+
\Pr(N\ge2)\Pr(i\in A(t+\delta)\mid i\notin A(t),N\ge2).
\end{align*}

If \(N = 0\), i.e., no event happened, then the current solution will not change and thus the first term is \(0\).
Furthermore, the probability of \(N \geq 2\) when
\(\delta \to 0^+\) is negligible.
Thus, only the middle term is left.
By \ref{it-j} from \Cref{def:valid-swap-pair}, we have that any element is added with probability at most \(1/k\). Hence,
\begin{align*}
&\left.
\frac{\partial}{\partial r}
\Pr(i\in A(r)\mid i\notin A(t))
\right|_{r=t}
=
\lambda(t)
\lim_{\delta\to0^+}
\Pr(i\in A(t+\delta)\mid i\notin A(t),N=1)
\le
\lambda(t)\cdot\frac1k.
\end{align*}
Thus,
\begin{equation}
\label{eq:pr-i-joins}
\left.
\frac{\partial}{\partial r}
\Pr(i\in A(r)\mid i\notin A(t))
\right|_{r=t}
\leq
\frac1{t}
\end{equation}

Next we analyze the probability of an element that is in the current set is not dropped, or equivalently, that it is maintained by the algorithm.
By definition,
\begin{align*}
&
\left.
\frac{\partial}{\partial r}
\Pr(i\in A(r)\mid i\in A(t))
\right|_{r=t}
=
\lim_{\delta\to0^+}
\frac1\delta
\Bigl(
\Pr(i\in A(t+\delta)\mid i\in A(t))
-1
\Bigr).
\end{align*}
Again, let $N$ denote the number of Poisson arrivals in $(t,t+\delta]$.
Then
\begin{align*}
\Pr(i\in A(t+\delta)\mid i\in A(t))
&=
\Pr(N=0)
\Pr(i\in A(t+\delta)\mid i\in A(t),N=0)
\\
&+
\Pr(N=1)
\Pr(i\in A(t+\delta)\mid i\in A(t),N=1)
\\
&+
\Pr(N\ge2)
\Pr(i\in A(t+\delta)\mid i\in A(t),N\ge2).
\end{align*}

If $N=0$, i.e., no arrival occurred, we have
\[
\Pr(i\in A(t+\delta)\mid i\in A(t),N=0)=1.
\]
And again, the probability that \(N \geq 2\) is negligible.
Therefore
\begin{align*}
&
\left.
\frac{\partial}{\partial r}
\Pr(i\in A(r)\mid i\in A(t))
\right|_{r=t}
=
\lambda(t)
\lim_{\delta\to0^+}
\Bigl(
\Pr(i\in A(t+\delta)\mid i\in A(t),N=1)-1
\Bigr).
\end{align*}
Condition on the unique arrival when $N=1$.
With probability $1-\frac1k$ we have \(I \neq i\) by \ref{it-i}, in which case $i$ remains in $A$.
With probability at most $\frac1k$ we have \(I = J =i\), in this case the item would remain with probability $1-t$. 
Hence
\begin{align*}
\Pr(i\in A(t+\delta)\mid i\in A(t),N=1)
&=
\Pr[I \neq i] + \Pr[I = J = i]\Pr[i \text{ survives drop}]\\
&\leq
\Pr[I \neq i] + \Pr[I = i]\Pr[i \text{ survives drop}]\\
&=
\left(1-\frac1k\right)
+
\frac1k\left(1-t\right)
=
1-\frac{t}{k}.
\end{align*}
Consequently,
\begin{equation}
\label{eq:pr-i-stays}
\left.
\frac{\partial}{\partial r}
\Pr(i\in A(r)\mid i\in A(t))
\right|_{r=t}
\le
\lambda(t)\left(-\frac{t}{k}\right)
=
\frac{k}{t}\left(-\frac{t}{k}\right)
=
-1.
\end{equation}

We can now analyze \(p_i(t)\).
It holds that
\begin{align*}
\left.
\frac{\partial p_i(r)}{\partial r}
\right|_{r=t}
&=
\left.
\frac{\partial}{\partial r}
\Pr(i\in A(r))
\right|_{r=t}
\\
&=
\left.
\frac{\partial}{\partial r}
\left[
\Pr(i\notin A(t))\Pr(i\in A(r)\mid i\notin A(t))
\right]
\right|_{r=t}\\
&
+
\left.
\frac{\partial}{\partial r}
\left[
\Pr(i\in A(t))\Pr(i\in A(r)\mid i\in A(t))
\right]
\right|_{r=t}
\\
&\leq  p_i(t) \cdot (-1) + (1-p_i(t)) \cdot \frac{1}{t} &&\text{by \eqref{eq:pr-i-joins} and \eqref{eq:pr-i-stays}}\\
&=\frac{1}{t} - p_i(t)\cdot \left(1+\frac{1}{t}\right)
\end{align*}

Thus, \(p_i(t)\) satisfies the previous differential equation and the starting condition that \(p_i(\varepsilon) = \varnothing\), since \(A(\varepsilon) = 0\).
Observe that
\begin{align*}
\left.
\frac{\partial (r e^r p_i(r))}{\partial r}
\right|_{r=t}
&=
e^t p_i(t) + te^tp_i(t) + te^t\left.\frac{\partial (p_i(r))}{\partial r}
\right|_{r=t}\\
&\leq
e^t p_i(t) + te^tp_i(t) + e^t - te^tp_i(t)
- e^tp_i(t)\\
&=
e^t.
\end{align*}
By integrating both sides from \(\varepsilon\)
to \(t\), we have
\[
t e^t p_i(t) \leq \int_\varepsilon^t e^s ds = e^t - e^\varepsilon.
\]
Dividing both sides by \(t e^t\), we get
\[
p_i(t) \leq \frac1{t} (1 - e^{-t + \varepsilon}).
\qedhere
\]
\end{proof}

\subsection{\texorpdfstring{Proof of \Cref{lem:almost_average_poisson}}{}}

We start by proving the analogous approximated version of \Cref{lem:potential_average}.

\begin{lemma}
\label{lem:almost_potential_average}   
If \Cref{alg:Poisson} uses an $\eta$-almost-above-average valid swap procedure, then for any \(t \in [\varepsilon, 1)\) and condition on the event \(A(t) = S\).
    it holds that
     \[
    \frac{\partial_+}{\partial r} V_{S, t}(r) \big\rvert_{r = t} \geq  F(t\1_S \vee \1_\OPT) - F(t\1_S) - \eta.
    \]
    Therefore, it holds that
     \[
    Q'(t) \geq \E\!
    \left[
    F(\1_{\OPT}\vee t\1_{A(t)})
    \right]
    -
    Q(t) - \eta.
    \]
\end{lemma}
\begin{proof}
     From \Cref{lem:cond-pot} and \eqref{eq:almost-above-avg},
    we immediately obtain
    \begin{equation}
    \label{eq:almost-partial-bound}
    \frac{\partial_+}{\partial r} V_{S, t}(r) \big\rvert_{r = t} \geq F(t\1_S \vee \1_\OPT) - F(t\1_S) - \eta.
    \end{equation}
    Furthermore,
    By the definition of \(Q(r)\) and \(V_{S, t}(r)\), for every $r \geq t$,
\[
Q(r)
=
\sum_{S\in\Ical}
\Pr(A(t)=S)\cdot V_{S,t}(r).
\]
Hence,
\begin{align*}
Q'(t)
&=
\sum_{S\in\Ical}
\Pr(A(t)=S)
\left.
\frac{\partial V_{S,t}(r)}
{\partial r}
\right|_{r=t}
\\
&\ge
\sum_{S\in\Ical}
\Pr(A(t)=S)
\left(F(\1_{\OPT}\vee t\1_S)-F(t\1_S) - \eta \right)
&&\text{by \eqref{eq:almost-partial-bound}}
\\
&=
\sum_{S\in\Ical}
\Pr(A(t)=S)
F(\1_{\OPT}\vee t\1_S)
-
\sum_{S\in\Ical}
\Pr(A(t)=S)
F(t\1_S)
-\eta
\sum_{S\in\Ical}
\Pr(A(t)=S)
\\
&=
\E\!\left[
F(\1_{\OPT}\vee t\1_{A(t)})
\right]
-
Q(t)
-\eta.
\end{align*}
    
\end{proof}

\begin{proof}[Proof of \Cref{lem:almost_average_poisson}]

We start with the case when \(f\) is a monotone submodular function.
By \Cref{lem:almost_potential_average}, we have
\[
Q'(r) \geq
\E\!\left[ F(\1_{\OPT}\vee r\1_{A(r)}) \right]
-
Q(r) - \eta
\geq
f(\OPT) - Q(r) - \eta,
\]
where in the second inequality we used monotonicity. We now have
\[
\frac{\partial e^r Q(r)}{\partial r}
=
e^r Q(r) + e^rQ'(r)
\geq
e^r f(\OPT) - e^r \eta.
\]
By integrating from \(\varepsilon\) to \(t\), we obtain \(e^t Q(t) - e^\varepsilon Q(\varepsilon)
\geq (e^t - e^\varepsilon)(f(\OPT) - \eta)\).
Using the fact that \(Q(\varepsilon) = f(\varnothing)\), and evaluating the expression 
at \(t = 1\), we get
\[
Q(1) \geq \left(1 - \frac{e^\varepsilon}{e}\right)(f(\OPT) -\eta) + \frac{e^\varepsilon}{e} f(\varnothing)
\geq
(1 - \varepsilon)(1 - \nicefrac{1}{e}) f(\OPT)
+ e^{-1} f(\varnothing) - \eta.
\]

Since \ref{it-valid} guarantees that \(A(t)\) is feasible with probability \(1\) for every \(t\), the set \(A(1)\) is a feasible solution with expected value \(Q(1)\).

We now look at the non-monotone case.
We again start with \Cref{lem:almost_potential_average}:
\[
Q'(r) \geq \E\!\left[ F(\1_{\OPT}\vee r\1_{A(r)}) \right]
- Q(r) -\eta.
\]
Instead of monotonicity,
we now apply \Cref{lem:lovasz-estimate}. Each element from \(\OPT\) appears with probability \(1\) and for each element not in \(\OPT\):
\[
\Pr[i \in R] = r \Pr[i \in A(r)] \leq 1 - e^{-r+\varepsilon}.
\]
Consequently,
\[
Q'(r) \geq e^{-r+\varepsilon} f(\OPT) -Q(r) -\eta.
\]
We now have
\[
\frac{\partial e^r Q(r)}{\partial r}
=
e^r Q(r) + e^rQ'(r)
\geq
e^\varepsilon f(\OPT) - e^r \eta.
\]
By integrating from \(\varepsilon\) to \(t\), we obtain \(e^t Q(t) - e^\varepsilon Q(\varepsilon)
\geq (t-\varepsilon)e^\varepsilon f(\OPT) - (e^t - e^\varepsilon)\eta\).
Using the fact that \(Q(\varepsilon) = f(\varnothing)\), and evaluating the expression 
at \(t = 1\), we get
\[
Q(1) \geq (1-\varepsilon)\left(\frac{e^\varepsilon}{e}\right)f(\OPT) - \left(1 - \frac{e^\varepsilon}{e}\right)\eta + \frac{e^\varepsilon}{e} f(\varnothing)
\geq
(1 - \varepsilon)(\nicefrac{1}{e}) f(\OPT)
+ e^{-1} f(\varnothing) - \eta.
\]
Since \ref{it-valid} guarantees that \(A(t)\) is feasible with probability \(1\) for every \(t\), the set \(A(1)\) is a feasible solution with expected value \(Q(1)\).

The expected number of swap procedure calls is the expected number of events in the Poisson process:
\[
\int_\varepsilon^1 \lambda(t) dt
=
k
\int_\varepsilon^1 \frac{1}{t} dt
=
k \ln(1/\varepsilon). \qedhere
\]
\end{proof}

\section{Preprocessing}\label{sec:preprocessing}

In this section, we prove Lemma~\ref{lema:imppre2}. 
Before giving the algorithm, we use the following lemma from \cite{BFNS14} that returns a constant factor estimate of the optimal value. 
\begin{lemma} \label{lem:estimate}\cite{BFNS14}
There is a polynomial time algorithm that returns a number $V$ such that $\max_{S\in \I} f(S)\leq V\leq \cons \max_{S\in \I} f(S)$ that makes $O(nk)$ queries for a general matroid and $O(n)$ queries for a partition matroid.   
\end{lemma}

Now, we give the algorithm for a general matroid~\ref{algo:advpre}. For this algorithm, we prove the Lemma~\ref{lema:imppre2} for $\delta=0$ since the running time does not degrade with $\delta$. Observe that the first property holds due to the stopping condition and $V\leq \cons \max_{S\in \I}f(S)$.
We now show the second property before discussing the runtime analysis.
Let $\tau $ denote the stopping time. Observe that $\tau\leq k$. Let $\F_t$ denote the filtration defined by the random choices of the algorithm until step $t$.
\begin{algorithm}[!h]
\caption{Advanced Preprocessing}
\label{algo:advpre}
\vspace{4pt}\hrule\vspace{4pt}

$Q_0 \gets \varnothing,\ i\gets 0$\;

\While{
$\displaystyle
\max_{T\subseteq U:\,T\cup Q_i\in\mathcal I}
\left\{
\sum_{j\in T}f(Q_i+j)-f(Q_i)
\right\}
\ge \constwo V$
}{
    \(
    Z \gets
    \arg\max_{T\subseteq U:\,T\cup Q_i\in\mathcal I}
    \left\{
    \sum_{j\in T}f(Q_i+j)-f(Q_i)
    \right\}
    \)
    
    \tcp{We can assume $Z\cap Q_i=\emptyset$ and that $Z\cup Q_i$ is a base}

    Let $j_i$ be a uniformly random element from the base $Z$\;

    $Q_{i+1}\gets Q_i+j_i$\;

    $i\gets i+1$\;
}

\Return{$\bar S\coloneqq Q_i$}\;

\end{algorithm}

Let $O$ denote an optimal solution. For every $1\leq t\leq \tau$, let $Z_t$ denote the set $Z$ selected in step $t$. Observe that $Z_t\cup Q_{t}\in \I$. We now show random subsets $O_t$ for each $t\leq \tau$ such that $O_t\subseteq O_{t-1}$ for each $t\leq \tau$ and $Q_t\cup O_t\in \I$. Observe that $O_0=O$ suffices since $Q_0=\emptyset$. Now we define $O_{t+1}$ for any $t\geq 0$. Let  $h_t:Z_t\cup Q_t\rightarrow O_t\cup Q_t$ be a map such that $O_t\cup Q_t\setminus {h(j)}\cup \{j\}\in \I$ for each  $j\in Z_t\cup Q_t$ and it is identity on the  common elements. We define $O_{t+1}\gets O_t\setminus h(j_t)$ where $j_t$ is the element chosen in step $t$ of the algorithm. Observe that the two conditions are satisfied by construction since $Q_{t+1}=Q_t\cup \{j_t\}$.

Let $M_t:=f(O_t\cup Q_t)+\frac12 f(Q_t)$ for each $t$. We now show that $M_{t\wedge \tau}$ is a sub-martingale. 
Indeed for any $t< \tau$, we have

\begin{align*}
    E[f(Q_{t+1})-f(Q_t)|\F_t]\geq \frac{\constwo V}{k-t}
\end{align*}
since $|Z|=k-t$.
While, for any $t\leq k-2$, we have
\begin{align}
    E[f(Q_{t+1}\cup O_{t+1})|\F_t]\geq \left(1-\frac{2}{k-t}\right) f(Q_t\cup O_t)\label{eqn:q_o_t}
\end{align}

since any element from $Q_t \cup O_t$ is dropped with probability at most $\frac{1}{k-t}$ and any element not in $Q_t \cup O_t$  is included in $Q_{t+1}\cup O_{t+1}$ with probability at most $\frac{1}{k-t}$. Thus the inequality follows from Lemma~\ref{lem:lovasz-estimate} by setting $A=Q_t\cup O_t$ and $p=1-\frac{1}{k-t}$ and $q=\frac{1}{k-t}$. Putting together we obtain,  

\begin{align*}
    E[f(Q_{t+1}\cup O_{t+1})+\frac12f(Q_{t+1})-f(Q_t\cup O_t)-\frac12f(Q_t)|\F_t]\geq \frac{-2}{k-t}f(Q_t\cup O_t) +\frac{10 V}{k-t}\geq \frac{8\cdot OPT}{k-t}  \geq 0 
\end{align*}
where we use $V\geq OPT$ and $f(Q_t\cup O_t)\leq OPT$. By optional stopping theorem, we have 
$$E[f(Q_{\tau}\cup O_{\tau})+\frac12f(Q_{\tau})]\geq f(Q_{0}\cup O_0] + \frac12f(Q_{0})$$
or equivalently, 
$$E[f(\bar{S}\cup O_{\tau})+\frac12 f(\bar{S})]\geq f(O) $$
where we use the fact that $Q_{\tau}=\bar{S}$, $Q_0=\emptyset$ and $O_{0}=O.$ Since $\bar{S}\cup O_{\tau}\in \I$, this proves the second property. 

Now, we analyze the run-time.  A straightforward implementation gives a $O(nk)$ time with $k$ iterations, each taking $O(n)$ time.

Now, we show how to give a faster implementation for a partition matroid with a slight loss in property (2).  First, we observe that, if we do not check the condition to exit, then the algorithm can first sample a random part and only restrict itself to elements in the same part. This takes $O(n)$ samples over all iterations to compute $Q_i$ for each $i=0,1,\ldots, k$. The main challenge is checking the criterion on each of the $k$ iterations which takes $O(n)$ time per iteration. We now show that it is enough to check the criterion every $\frac{\delta k}{20}$ iterations, and output the first time the condition is violated. Observe that we check total of $O(\frac{1}{\delta})$ times, where each check takes time $O(n)$ and thus the number of oracle calls is $O(\frac{n}{\delta})$ as claimed. Let $\sigma\geq \tau$ be the first time in the list that we check where $\tau$ was the first time where the exit condition is satisfied. Clearly, $\sigma-\tau<\frac{\delta k}{20}$.  

Observe that $Q_{\sigma}$ satisfies the first condition as the LHS only decreases over iterations. We show that the second condition is satisfied with an error of $\delta.$ 

Recall that for $t< \tau$, we have
$$ E[M_{t+1}-M_{t}|\F_{t}]\geq \frac{8\cdot OPT}{k-t}. $$
Let $\phi(s)=\sum_{t=0}^{s-1} \frac{1}{k-t}$. Then 
$M_t-8\phi(t)OPT$ is a submartingale upto time $\tau$. Thus by optional stopping theorem, 
$$ E[M_{\tau}-8\phi(\tau)]\geq  M_0-8\phi(0)$$
or equivalently,
$$ E[M_{\tau}]\geq 8E[\phi(\tau)] +OPT$$
 We now bound the loss from $M_{\sigma}$ to $M_{\tau}$. We first observe that $f(Q_{\sigma})\geq f(Q_{\tau})$ as $f(Q_{t+1})\geq f(Q_t)$ for each $t$, since the algorithm can always select a dummy element with zero marginal.

Thus, 
\begin{equation}\label{eqn:bound1} M_{\sigma}-M_{\tau}=f(Q_{\sigma}\cup O_{\sigma})+\frac12 f(Q_{\sigma})-f(Q_{\tau}\cup O_{\tau})-\frac12 f(Q_{\tau})\geq f(Q_{\sigma}\cup O_{\sigma})-f(Q_{\tau}\cup O_{\tau})\geq -OPT 
\end{equation}

where the last inequality uses that $Q_{\tau}\cup O_{\tau}$ is feasible. 
Moreover, for any $t$, we have
\begin{align}
E[M_{t+1}-M_t|\F_t]&=E[f(Q_{t+1}\cup O_{t+1})+\frac12 f(Q_{t+1})-f(Q_{t}\cup O_{t})-\frac12 f(Q_{t})]\nonumber\\
&\geq E[f(Q_{t+1}\cup O_{t+1})-f(Q_{t}\cup O_{t})]\nonumber\\
&\geq -\frac{2}{k-t} OPT \label{eqn:bound2}
\end{align} 
where the last inequality follows from \eqref{eqn:q_o_t}.
Thus $M_t+2 \phi(t) OPT$ is a submartingale and by optional stopping theorem, we have
\begin{align*}
E[M_{\sigma}-M_{\tau}|\F_{\tau}]\geq -2E[\phi(\sigma)-\phi(\tau)|\F_{\tau}] OPT
\end{align*}

Now we show,
$$ E[8\phi(t)+M_\sigma-M_\tau|\F_{\tau}]\geq -\delta OPT$$
We make two cases depending on the event $A=\{\phi(\tau)\geq \frac18\}$ that is $\F_\tau$-measurable.

\textbf{Case 1.} $A$ happens. Then  
\begin{align*}
    E[8\phi(\tau)+M_\sigma-M_\tau|\F_{\tau}]\geq E[8\phi(\tau)-OPT|\F_{\tau}]\geq 0
\end{align*}
as claimed. 

\textbf{Case 2.} $A$ doesn't happen. Since,
Observe $\phi(\tau)\geq  \frac{\tau}{k}$ and thus by the assumption, we have $\tau\leq \frac{k}{8}$ if $A$ doesn't happen.  Then
\begin{align*}
    \phi(\sigma)-\phi(\tau)=\sum_{t=\tau+1}^\sigma \frac{1}{k-t}\leq \frac{\delta k/20}{k-k/8-\delta k/20}\leq \frac{\delta}{10}
\end{align*}
and thus
\begin{align*}
    E[8\phi(t)+M_\sigma-M_\tau|\F_{\tau}]\geq  -2E[\phi(\sigma)-\phi(\tau)|\F_{\tau}] OPT\geq  -\frac{\delta}{5} OPT
\end{align*}

Thus putting together, we have
\begin{align*}
    \E[M_\sigma]&=\E[M_\tau]+\E[M_\sigma-M_\tau]\\
     &\geq OPT + \E[ 8\phi(\tau)]+\E[M_{\sigma}-M_\tau]\\
    &\geq OPT + \E[ 8\phi(\tau)+M_{\sigma}-M_\tau|\F_{\tau}]]\\
     &\geq OPT +\E[-\delta OPT]\\
     &\geq (1-\delta)OPT
\end{align*}

and thus 
 $E[M_{\sigma}]\geq (1-\delta) OPT$ as required. 

{\small
\printbibliography

@article{BuchcinderF16a,
  title={Constrained Submodular Maximization via a Nonsymmetric Technique},
  author={Buchbinder, Niv and Feldman, Moran},
  journal={Mathematics of Operations Research},
  volume={44},
  number={3},
  pages={988--1005},
  year={2019},
  publisher={INFORMS}
}

@book{schrijver2003combinatorial,
  title={Combinatorial optimization: polyhedra and efficiency},
  author={Schrijver, Alexander and others},
  volume={24},
  number={2},
  year={2003},
  publisher={Springer}
}

@inproceedings{BuchbinderF23a,
  author    = {Niv Buchbinder and
               Moran Feldman},
  title     = {Constrained Submodular Maximization via New Bounds for DR-Submodular Functions},
  booktitle = {Proceedings of the 56th Annual {ACM} Symposium on Theory of Computing,
               {STOC} 2024, Vancouver, BC, Canada, June 24-28, 2024},
  pages     = {1820--1831},
  publisher = {ACM},
  year      = {2024}
}

@misc{BuchbinderF24a,
      title={Extending the Extension: Deterministic Algorithm for Non-monotone Submodular Maximization}, 
      author={Niv Buchbinder and Moran Feldman},
      year={2024},
      eprint={2409.14325},
      archivePrefix={arXiv},
      primaryClass={cs.DS},
      url={https://arxiv.org/abs/2409.14325}, 
}

@article{BuchbinderFS14a,
  title={Comparing apples and oranges: Query trade-off in submodular maximization},
  author={Buchbinder, Niv and Feldman, Moran and Schwartz, Roy},
  journal={Mathematics of Operations Research},
  volume={42},
  number={2},
  pages={308--329},
  year={2017},
  publisher={INFORMS}
}

@inproceedings{EneN16a,
  title={Constrained Submodular Maximization: Beyond 1/e},
  author={Ene, Alina and Nguyen, Huy L.},
  booktitle={Proceedings of the 57th Annual IEEE Symposium on Foundations of Computer Science (FOCS)},
  pages={248--257},
  year={2016},
  organization={IEEE}
}

@misc{HanCC20a,
      title={Deterministic Approximation for Submodular Maximization over a Matroid in Nearly Linear Time}, 
      author={Kai Han and Zongmai Cao and Shuang Cui and Benwei Wu},
      year={2020},
      eprint={2010.11420},
      archivePrefix={arXiv},
      primaryClass={cs.DS},
      url={https://arxiv.org/abs/2010.11420}, 
}

@misc{GanzKSS26a,
      title={A Poisson Process for Submodular Maximization}, 
      author={Amit Ganz-Rozenman and Ariel Kulik and Roy Schwartz and Mohit Singh},
      year={2026},
      eprint={2605.03071},
      archivePrefix={arXiv},
      primaryClass={cs.DS},
      url={https://arxiv.org/abs/2605.03071}, 
}

@misc{Segui-GascoS17a,
      title={Fast Non-Monotone Submodular Maximisation Subject to a Matroid Constraint}, 
      author={Pau Segui-Gasco and Hyo-Sang Shin},
      year={2017},
      eprint={1703.06053},
      archivePrefix={arXiv},
      primaryClass={cs.DS},
      url={https://arxiv.org/abs/1703.06053}, 
}

@InProceedings{Sakaue20a,
  title = 	 {Guarantees of Stochastic Greedy Algorithms for Non-monotone Submodular Maximization with Cardinality Constraints},
  author =       {Sakaue, Shinsaku},
  booktitle = 	 {Proceedings of the Twenty Third International Conference on Artificial Intelligence and Statistics},
  pages = 	 {11--21},
  year = 	 {2020},
  volume = 	 {108},
  series = 	 {Proceedings of Machine Learning Research},
  month = 	 {26--28 Aug},
  publisher =    {PMLR}
}

@inproceedings{TukanMF24a,
 author = {Tukan, Murad and Mualem, Loay and Feldman, Moran},
 booktitle = {Advances in Neural Information Processing Systems},
 pages = {51223--51253},
 publisher = {Curran Associates, Inc.},
 title = {Practical 0.385-Approximation for Submodular Maximization Subject to a Cardinality Constraint},
 volume = {37},
 year = {2024}
}

@book{ross2014introduction,
  title={Introduction to Probability Models},
  author={Ross, Sheldon M.},
  year={2014},
  edition={11th},
  publisher={Academic Press},
  isbn={9780124079489}
}

@article{fisher1978analysis,
author="Fisher, M. L.
and Nemhauser, G. L.
and Wolsey, L. A.",
editor="Balinski, M. L.
and Hoffman, A. J.",
title="An analysis of approximations for maximizing submodular set functions---{II}",
journal="Math. Prog. Study.",
year="1978",
pages="73--87",
volume="8"
}

@article{nemhauser1978analysis,
  title={An analysis of approximations for maximizing submodular set functions—I},
  author={Nemhauser, George L and Wolsey, Laurence A and Fisher, Marshall L},
  journal={Mathematical programming},
  volume={14},
  pages={265--294},
  year={1978},
  publisher={Springer}
}

@article{NW78,
  title={Best algorithms for approximating the maximum of a submodular set function},
  author={Nemhauser, George L and Wolsey, Laurence A},
  journal={Mathematics of operations research},
  volume={3},
  number={3},
  pages={177--188},
  year={1978},
  publisher={INFORMS}
}

@inproceedings{kempe2003maximizing,
  title={Maximizing the spread of influence through a social network},
  author={Kempe, David and Kleinberg, Jon and Tardos, {\'E}va},
  booktitle={Proceedings of the ninth ACM SIGKDD international conference on Knowledge discovery and data mining},
  pages={137--146},
  year={2003}
}

@inproceedings{HMS08,
author = {Hartline, Jason and Mirrokni, Vahab and Sundararajan, Mukund},
title = {Optimal marketing strategies over social networks},
year = {2008},
isbn = {9781605580852},
publisher = {Association for Computing Machinery},
address = {New York, NY, USA},
url = {https://doi.org/10.1145/1367497.1367524},
doi = {10.1145/1367497.1367524},
booktitle = {Proceedings of the 17th International Conference on World Wide Web},
pages = {189–198},
numpages = {10},
location = {Beijing, China},
series = {WWW '08}
}

@article{MR10,
author = {Mossel, Elchanan and Roch, Sebastien},
title = {Submodularity of Influence in Social Networks: From Local to Global},
journal = {SIAM Journal on Computing},
volume = {39},
number = {6},
pages = {2176-2188},
year = {2010},
doi = {10.1137/080714452},
URL = { 
    
        https://doi.org/10.1137/080714452
},
eprint = { 
        https://doi.org/10.1137/080714452
}
}

@article{KLGVF08,
author = {Andreas Krause  and Jure Leskovec  and Carlos Guestrin  and Jeanne VanBriesen  and Christos Faloutsos },
title = {Efficient Sensor Placement Optimization for Securing Large Water Distribution Networks},
journal = {Journal of Water Resources Planning and Management},
volume = {134},
number = {6},
pages = {516-526},
year = {2008},
doi = {10.1061/(ASCE)0733-9496(2008)134:6(516)},
}

@article{KSG08,
author = {Krause, Andreas and Singh, Ajit and Guestrin, Carlos},
title = {Near-Optimal Sensor Placements in Gaussian Processes: Theory, Efficient Algorithms and Empirical Studies},
year = {2008},
issue_date = {6/1/2008},
publisher = {JMLR.org},
volume = {9},
issn = {1532-4435},
journal = {J. Mach. Learn. Res.},
month = jun,
pages = {235–284},
numpages = {50}
}

@inproceedings{LB10,
author = {Lin, Hui and Bilmes, Jeff},
title = {Multi-document summarization via budgeted maximization of submodular functions},
year = {2010},
isbn = {1932432655},
publisher = {Association for Computational Linguistics},
address = {USA},
booktitle = {Human Language Technologies: The 2010 Annual Conference of the North American Chapter of the Association for Computational Linguistics},
pages = {912–920},
numpages = {9},
location = {Los Angeles, California},
series = {HLT '10}
}

@inproceedings{LB11,
    title = "A Class of Submodular Functions for Document Summarization",
    author = "Lin, Hui  and
      Bilmes, Jeff",
    editor = "Lin, Dekang  and
      Matsumoto, Yuji  and
      Mihalcea, Rada",
    booktitle = "Proceedings of the 49th Annual Meeting of the Association for Computational Linguistics: Human Language Technologies",
    month = jun,
    year = "2011",
    address = "Portland, Oregon, USA",
    publisher = "Association for Computational Linguistics",
    url = "https://aclanthology.org/P11-1052/",
    pages = "510--520"
}

@inproceedings{mirza16,
  title={Fast constrained submodular maximization: Personalized data summarization},
  author={Mirzasoleiman, Baharan and Badanidiyuru, Ashwinkumar and Karbasi, Amin},
  booktitle={International Conference on Machine Learning},
  pages={1358--1367},
  year={2016},
  organization={PMLR}
}

@inproceedings{LWKSB13,
author = {Liu, Yuzong and Wei, Kai and Kirchhoff, Katrin and Song, Yisong and Bilmes, Jeff},
year = {2013},
pages = {7184-7188},
title = {Submodular feature selection for high-dimensional acoustic score spaces},
booktitle = {Acoustics, Speech, and Signal Processing, 1988. ICASSP-88., 1988 International Conference on},
doi = {10.1109/ICASSP.2013.6639057}
}

@InProceedings{KEDNG17,
  title = 	 {{Scalable Greedy Feature Selection via Weak Submodularity}},
  author = 	 {Khanna, Rajiv and Elenberg, Ethan and Dimakis, Alex and Negahban, Sahand and Ghosh, Joydeep},
  booktitle = 	 {Proceedings of the 20th International Conference on Artificial Intelligence and Statistics},
  pages = 	 {1560--1568},
  year = 	 {2017},
  editor = 	 {Singh, Aarti and Zhu, Jerry},
  volume = 	 {54},
  series = 	 {Proceedings of Machine Learning Research},
  month = 	 {20--22 Apr},
  publisher =    {PMLR},
  url = 	 {https://proceedings.mlr.press/v54/khanna17b.html}
}

@inproceedings{BHZ22,
author = {Bao, Wei-Xuan and Hang, Jun-Yi and Zhang, Min-Ling},
title = {Submodular Feature Selection for Partial Label Learning},
year = {2022},
isbn = {9781450393850},
publisher = {Association for Computing Machinery},
address = {New York, NY, USA},
url = {https://doi.org/10.1145/3534678.3539292},
doi = {10.1145/3534678.3539292},
booktitle = {Proceedings of the 28th ACM SIGKDD Conference on Knowledge Discovery and Data Mining},
pages = {26–34},
numpages = {9},
location = {Washington DC, USA},
series = {KDD '22}
}

@book{B13,
author = {Bach, Francis},
title = {Learning with Submodular Functions: A Convex Optimization Perspective},
year = {2013},
isbn = {1601987560},
publisher = {Now Publishers Inc.},
address = {Hanover, MA, USA},
}

@article{CVZ14,
author = {Chekuri, Chandra and Vondr\'{a}k, Jan and Zenklusen, Rico},
title = {Submodular Function Maximization via the Multilinear Relaxation and Contention Resolution Schemes},
journal = {SIAM Journal on Computing},
volume = {43},
number = {6},
pages = {1831-1879},
year = {2014},
doi = {10.1137/110839655},
URL = {
        https://doi.org/10.1137/110839655
},
eprint = { 
        https://doi.org/10.1137/110839655
}
}

@inproceedings{FNS11,
author = {Feldman, Moran and Naor, Joseph (Seffi) and Schwartz, Roy},
title = {A Unified Continuous Greedy Algorithm for Submodular Maximization},
year = {2011},
isbn = {9780769545714},
publisher = {IEEE Computer Society},
address = {USA},
url = {https://doi.org/10.1109/FOCS.2011.46},
doi = {10.1109/FOCS.2011.46},
booktitle = {Proceedings of the 2011 IEEE 52nd Annual Symposium on Foundations of Computer Science},
pages = {570–579},
numpages = {10},
series = {FOCS '11}
}

@article{LMNS10,
author = {Lee, Jon and Mirrokni, Vahab S. and Nagarajan, Viswanath and Sviridenko, Maxim},
title = {Maximizing Nonmonotone Submodular Functions under Matroid or Knapsack Constraints},
journal = {SIAM Journal on Discrete Mathematics},
volume = {23},
number = {4},
pages = {2053-2078},
year = {2010}
}

@inproceedings{BFNS14,
	title={Submodular maximization with cardinality constraints},
	author={Buchbinder, Niv and Feldman, Moran and Naor, Joseph and Schwartz, Roy},
	booktitle={Proceedings of the twenty-fifth annual ACM-SIAM symposium on Discrete algorithms},
	pages={1433--1452},
	year={2014},
	organization={SIAM}
}

@InProceedings{GNS23,
  author =	{Ganz, Amit and Nuti, Pranav and Schwartz, Roy},
  title =	{{A Tight Competitive Ratio for Online Submodular Welfare Maximization}},
  booktitle =	{31st Annual European Symposium on Algorithms (ESA 2023)},
  pages =	{52:1--52:17},
  year =	{2023},
  volume =	{274}
}

@inproceedings{OV11,
  title={Submodular Maximization by Simulated Annealing},
  author={Oveis Gharan, Shayan and Vondr{\'a}k, Jan},
  booktitle={Proceedings of the Twenty-Second Annual ACM-SIAM Symposium on Discrete Algorithms},
  pages={1098--1116},
  year={2011},
  organization={SIAM},
  doi={10.1137/1.9781611973082.83}
}

@article{V13,
author = {Vondr\'{a}k, Jan},
title = {Symmetry and Approximability of Submodular Maximization Problems},
journal = {SIAM Journal on Computing},
volume = {42},
number = {1},
pages = {265-304},
year = {2013}
}

@article{CCPV11,
  title={Maximizing a monotone submodular function subject to a matroid constraint},
  author={Calinescu, Gruia and Chekuri, Chandra and Pal, Martin and Vondr{\'a}k, Jan},
  journal={SIAM Journal on Computing},
  volume={40},
  number={6},
  pages={1740--1766},
  year={2011},
  publisher={SIAM}
}

@article{filmus2014monotone,
  title={Monotone submodular maximization over a matroid via non-oblivious local search},
  author={Filmus, Yuval and Ward, Justin},
  journal={SIAM Journal on Computing},
  volume={43},
  number={2},
  pages={514--542},
  year={2014},
  publisher={SIAM}
}

@inproceedings{BF24deter,
  title={Deterministic algorithm and faster algorithm for submodular maximization subject to a matroid constraint},
  author={Buchbinder, Niv and Feldman, Moran},
  booktitle={2024 IEEE 65th Annual Symposium on Foundations of Computer Science (FOCS)},
  pages={700--712},
  year={2024},
  organization={IEEE}
}

@Article{AS04,
author="Ageev, A.A.
and Sviridenko, M.I.",
title="Pipage Rounding: A New Method of Constructing Algorithms with Proven Performance Guarantee",
journal="Journal of Combinatorial Optimization",
year="2004",
month="9",
day="01",
volume="8",
number="3",
pages="307--328"
}

@article{AS99,
author = {Ageev, A. A. and Sviridenko, M. I.},
title = {An 0.828–approximation algorithm for the uncapacitated facility location problem},
year = {1999},
issue_date = {July 20, 1999},
volume = {93},
number = {2–3},
journal = {Discrete Appl. Math.},
month = jul,
pages = {149–156},
numpages = {8}
}

@article{CFN77,
author = {Cornuejols, Gerard and Fisher, Marshall L. and Nemhauser, George L.},
title = {Location of Bank Accounts to Optimize Float: An Analytic Study of Exact and Approximate Algorithms},
year = {1977},
issue_date = {April 1977},
volume = {23},
number = {8},
journal = {Manage. Sci.},
month = apr,
pages = {789–810},
numpages = {22}
}

@incollection{CFN77b,
title = {On the Uncapacitated Location Problem},
series = {Annals of Discrete Mathematics},
publisher = {Elsevier},
volume = {1},
pages = {163-177},
year = {1977},
booktitle = {Studies in Integer Programming},
author = {Gerard Cornuejols and Marshall Fisher and George L. Nemhauser}
}

@article{GW95,
author = {Goemans, Michel X. and Williamson, David P.},
title = {Improved approximation algorithms for maximum cut and satisfiability problems using semidefinite programming},
year = {1995},
issue_date = {Nov. 1995},
volume = {42},
number = {6},
journal = {J. ACM},
month = nov,
pages = {1115–1145},
numpages = {31}
}

@inproceedings{Karp72,
  author = {Karp, Richard M.},
  booktitle = {Complexity of Computer Computations},
  editor = {Miller, Raymond E. and Thatcher, James W.},
  pages = {85-103},
  publisher = {Plenum Press, New York},
  series = {The IBM Research Symposia Series},
  title = {Reducibility Among Combinatorial Problems.},
  year = {1972}
}

@article{KKMO07,
author = { Khot, Subhash and  Kindler, Guy and  Mossel, Elchanan and  O’Donnell, Ryan},
title = {Optimal Inapproximability Results for MAX‐CUT and Other 2‐Variable CSPs?},
journal = {SIAM Journal on Computing},
volume = {37},
number = {1},
pages = {319-357},
year = {2007}
}

@inproceedings{HZ01,
author = {Halperin, Eran and Zwick, Uri},
title = {Combinatorial approximation algorithms for the maximum directed cut problem},
year = {2001},
pages = {1–7},
numpages = {7},
series = {SODA '01}
}

@inproceedings{LLZ02,
author = {Lewin, Michael and Livnat, Dror and Zwick, Uri},
title = {Improved Rounding Techniques for the MAX 2-SAT and MAX DI-CUT Problems},
year = {2002},
booktitle = {Proceedings of the 9th International IPCO Conference on Integer Programming and Combinatorial Optimization},
pages = {67–82},
numpages = {16}
}

@article{BHPZ26,
author = {Brakensiek, Joshua and Huang, Neng and Potechin, Aaron and Zwick, Uri},
title = {Separating MAX 2-AND, MAX DI-CUT, and MAX CUT},
journal = {SIAM Journal on Computing},
volume = {55},
number = {3},
year = {2026}
}

@inproceedings{fleischer2006tight,
  title={Tight approximation algorithms for maximum general assignment problems},
  author={Fleischer, Lisa and Goemans, Michel X and Mirrokni, Vahab S and Sviridenko, Maxim},
  booktitle={SODA},
  volume={6},
  pages={611--620},
  year={2006}
}

@inproceedings{feige2006approximation,
  title={Approximation algorithms for allocation problems: Improving the factor of 1-1/e},
  author={Feige, Uriel and Vondr{\'a}k, Jan},
  booktitle={2006 47th Annual IEEE Symposium on Foundations of Computer Science (FOCS'06)},
  pages={667--676},
  year={2006},
  organization={IEEE}
}

@inproceedings{BS18,
author = {Balkanski, Eric and Singer, Yaron},
title = {The adaptive complexity of maximizing a submodular function},
year = {2018},
booktitle = {Proceedings of the 50th Annual ACM SIGACT Symposium on Theory of Computing},
pages = {1138–1151},
numpages = {14},
location = {Los Angeles, CA, USA},
series = {STOC 2018}
}

@inproceedings{BRS19,
author = {Balkanski, Eric and Rubinstein, Aviad and Singer, Yaron},
title = {An optimal approximation for submodular maximization under a matroid constraint in the adaptive complexity model},
year = {2019},
booktitle = {Proceedings of the 51st Annual ACM SIGACT Symposium on Theory of Computing},
pages = {66–77},
numpages = {12},
series = {STOC 2019}
}

@inproceedings{CQ19,
author = {Chekuri, Chandra and Quanrud, Kent},
title = {Parallelizing greedy for submodular set function maximization in matroids and beyond},
year = {2019},
booktitle = {Proceedings of the 51st Annual ACM SIGACT Symposium on Theory of Computing},
pages = {78–89},
numpages = {12},
series = {STOC 2019}
}

@inproceedings{CK19b,
  author    = {Chandra Chekuri and
               Kent Quanrud},
  title     = {Submodular Function Maximization in Parallel via the Multilinear Relaxation},
  booktitle = {Proceedings of the Thirtieth Annual {ACM-SIAM} Symposium on Discrete Algorithms},
  pages     = {303--322},
  publisher = {SIAM},
  year      = {2019}
}

@INPROCEEDINGS{DENW16,
  author={Da Ponte Barbosa, Rafael and Ene, Alina and Nguyen, Huy L. and Ward, Justin},
  booktitle={2016 IEEE 57th Annual Symposium on Foundations of Computer Science (FOCS)}, 
  title={A New Framework for Distributed Submodular Maximization}, 
  year={2016},
  pages={645-654}
}

@article{F98,
author = {Feige, Uriel},
title = {A Threshold of $\ln{n}$ for Approximating Set Cover},
year = {1998},
issue_date = {July 1998},
publisher = {Association for Computing Machinery},
address = {New York, NY, USA},
volume = {45},
number = {4},
journal = {J. ACM},
month = {7},
pages = {634–652},
numpages = {19}
}

@article{KMN99,
title = {The budgeted maximum coverage problem},
journal = {Information Processing Letters},
volume = {70},
number = {1},
pages = {39-45},
year = {1999},
author = {Samir Khuller and Anna Moss and Joseph (Seffi) Naor}
}

@incollection{Mc98,
  title={Concentration},
  author={McDiarmid, Colin},
  booktitle={Probabilistic methods for algorithmic discrete mathematics},
  pages={195--248},
  year={1998},
  publisher={Springer}
}

@article{FZ18,
author = {Feldman, Moran and Zenklusen, Rico},
title = {The Submodular Secretary Problem Goes Linear},
journal = {SIAM Journal on Computing},
volume = {47},
number = {2},
pages = {330-366},
year = {2018}
}

@article{BFS19online,
author = {Buchbinder, Niv and Feldman, Moran and Schwartz, Roy},
title = {Online Submodular Maximization with Preemption},
year = {2019},
issue_date = {July 2019},
publisher = {Association for Computing Machinery},
address = {New York, NY, USA},
volume = {15},
number = {3},
journal = {ACM Trans. Algorithms},
month = jun,
articleno = {30},
numpages = {31}
}

@article{BFFG20,
author = {Buchbinder, Niv and Feldman, Moran and Filmus, Yuval and Garg, Mohit},
title = {Online submodular maximization: beating 1/2 made simple},
year = {2020},
journal = {Mathematical Programming},
volume = {183},
issue = {1},
pages = {149-169}
}

@article{BHZ13,
author = {Bateni, Mohammadhossein and Hajiaghayi, Mohammadtaghi and Zadimoghaddam, Morteza},
title = {Submodular secretary problem and extensions},
year = {2013},
volume = {9},
number = {4},
journal = {ACM Trans. Algorithms},
month = oct,
articleno = {32},
numpages = {23}
}

@inproceedings{KPV13,
author = {Kapralov, Michael and Post, Ian and Vondr\'{a}k, Jan},
title = {Online submodular welfare maximization: greedy is optimal},
year = {2013},
booktitle = {Proceedings of the Twenty-Fourth Annual ACM-SIAM Symposium on Discrete Algorithms},
pages = {1216–1225},
numpages = {10},
series = {SODA '13}
}

@article{KMZ18,
author = {Korula, Nitish and Mirrokni, Vahab and Zadimoghaddam, Morteza},
title = {Online Submodular Welfare Maximization: Greedy Beats 1/2 in Random Order},
journal = {SIAM Journal on Computing},
volume = {47},
number = {3},
pages = {1056-1086},
year = {2018}
}

@article{AEFNS22,
author = {Alaluf, Naor and Ene, Alina and Feldman, Moran and Nguyen, Huy and Suh, Andrew},
title = {An Optimal Streaming Algorithm for Submodular Maximization with a Cardinality Constraint},
journal = {Mathematics of Operations Research},
volume = {47},
number = {4},
year = {2022}
}

@InProceedings{CGQ15,
author="Chekuri, Chandra
and Gupta, Shalmoli
and Quanrud, Kent",
title="Streaming Algorithms for Submodular Function Maximization",
booktitle="42nd International Colloquium on Automata, Languages, and Programming (ICALP)",
year="2015",
pages="318--330"
}

@InProceedings{FLNSZ22,
  author =	{Feldman, Moran and Liu, Paul and Norouzi-Fard, Ashkan and Svensson, Ola and Zenklusen, Rico},
  title =	{{Streaming Submodular Maximization Under Matroid Constraints}},
  booktitle =	{49th International Colloquium on Automata, Languages, and Programming (ICALP)},
  pages =	{59:1--59:20},
  year =	{2022},
}

@article{FNSZ23,
author = {Feldman, Moran and Norouzi-Fard, Ashkan and Svensson, Ola and Zenklusen, Rico},
title = {The One-Way Communication Complexity of Submodular Maximization with Applications to Streaming and Robustness},
year = {2023},
issue_date = {August 2023},
volume = {70},
number = {4},
journal = {J. ACM},
month = aug,
articleno = {24},
numpages = {52}
}

@InProceedings{KMZLK19,
  title = 	 {Submodular Streaming in All Its Glory: Tight Approximation, Minimum Memory and Low Adaptive Complexity},
  author =       {Kazemi, Ehsan and Mitrovic, Marko and Zadimoghaddam, Morteza and Lattanzi, Silvio and Karbasi, Amin},
  booktitle = 	 {Proceedings of the 36th International Conference on Machine Learning},
  pages = 	 {3311--3320},
  year = 	 {2019},
  volume = 	 {97},
  series = 	 {Proceedings of Machine Learning Research},
  month = 	 {09--15 Jun}
}

@misc{CGLXZ26,
      title={Deterministic Algorithm for Non-monotone Submodular Maximization under Matroid and Knapsack Constraints}, 
      author={Shengminjie Chen and Yiwei Gao and Kaifeng Lin and Xiaoming Sun and Jialin Zhang},
      year={2026},
      eprint={2603.11996},
      archivePrefix={arXiv},
      primaryClass={cs.DS}
}

@inproceedings{EN19,
  title={Towards Nearly-linear Time Algorithms for Submodular Maximization with a Matroid Constraint},
  author={Ene, Alina and Nguyen, Huy L},
  booktitle={International Colloquium on Automata, Languages, and Programming (ICALP)},
  volume={132},
  year={2019}
}

@book{Gonzalez2018HandbookV1,
  editor    = {Teofilo F. Gonzalez},
  title     = {Handbook of Approximation Algorithms and Metaheuristics, Second Edition: Methodologies and Traditional Applications},
  publisher = {Chapman and Hall/CRC},
  year      = {2018},
  isbn      = {9781482258905},
  url       = {https://routledge.com}
}

@article{Qi23,
author = {Qi, Benjamin},
title = {On Maximizing Sums of Non-monotone Submodular and Linear Functions},
year = {2023},
issue_date = {Apr 2024},
publisher = {Springer-Verlag},
address = {Berlin, Heidelberg},
volume = {86},
number = {4},
journal = {Algorithmica},
month = nov,
pages = {1080–1134},
numpages = {55}
}

@article{BFNS15,
  title={A tight linear time (1/2)-approximation for unconstrained submodular maximization},
  author={Buchbinder, Niv and Feldman, Moran and Naor, Joseph (Seffi) and Schwartz, Roy},
  journal={SIAM Journal on Computing},
  volume={44},
  number={5},
  pages={1384--1402},
  year={2015},
  publisher={SIAM}
}

@article{FMV11,
author = {Feige, Uriel and Mirrokni, Vahab S. and Vondr\'{a}k, Jan},
title = {Maximizing Non-monotone Submodular Functions},
year = {2011},
issue_date = {July 2011},
publisher = {Society for Industrial and Applied Mathematics},
address = {USA},
volume = {40},
number = {4},
issn = {0097-5397},
url = {https://doi.org/10.1137/090779346},
doi = {10.1137/090779346},
journal = {SIAM J. Comput.},
month = jul,
pages = {1133–1153},
numpages = {21}
}

@inproceedings{MSS12,
author = {Makarychev, Konstantin and Schudy, Warren and Sviridenko, Maxim},
title = {Concentration inequalities for nonlinear matroid intersection},
year = {2012},
booktitle = {Proceedings of the Twenty-Third Annual ACM-SIAM Symposium on Discrete Algorithms},
pages = {420–436},
numpages = {17},
location = {Kyoto, Japan},
series = {SODA '12}
}

@misc{BFLS26,
      title={Semi-Streaming Algorithms for Submodular Maximization under Random Arrival Order}, 
      author={Niv Buchbinder and Moran Feldman and Siyue Liu and Sherry Sarkar},
      year={2026},
      eprint={2605.14296},
      archivePrefix={arXiv},
      primaryClass={cs.DS},
      url={https://arxiv.org/abs/2605.14296}, 
}
}

\end{document}